\documentclass[11pt]{article}

\usepackage[utf8]{inputenc}
\usepackage[english]{babel}
\usepackage{graphicx}         
\usepackage[pdftex]{hyperref}
\usepackage{braket}
\usepackage{amsmath}
\usepackage{amssymb}
\usepackage{amsthm}
\usepackage{mathtools}
\usepackage{bbm}		
\usepackage{array}		
\usepackage{color}
\usepackage[usenames,dvipsnames]{xcolor}
\usepackage{caption}
\usepackage{enumitem}
\usepackage{subeqnarray}	
\usepackage{wrapfig}
\usepackage{cancel}
\usepackage{enumitem}

\allowdisplaybreaks

\numberwithin{equation}{section}

\theoremstyle{plain}\newtheorem{definition}{Definition}[section]
\newtheorem{lem}[definition]{Lemma}
\newtheorem{proposition}[definition]{Proposition}

\theoremstyle{remark}\newtheorem{remark}[definition]{Remark}

\theoremstyle{plain}

\theoremstyle{plain}\newtheorem{assumption}{Assumption}

\theoremstyle{plain}\newtheorem{theorem}{Theorem}

\newcommand{\thmit}[1]{\begin{enumerate}[label={(\alph*)}, ref={\thetheorem\alph*}]{#1}\end{enumerate}}
\newcommand{\lemit}[1]{\begin{enumerate}[label={(\alph*)}, ref={\thelem\alph*}]{#1}\end{enumerate}}

\newcommand{\defit}[1]{\begin{enumerate}[label={(\alph*)}, ref={\thedefinition\alph*}]{#1}\end{enumerate}}

\newcommand{\D}{\mathcal{D}}
\newcommand{\R}{\mathbb{R}}
\newcommand{\C}{\mathbb{C}}
\newcommand{\N}{\mathbb{N}}
\newcommand{\fH}{\mathfrak{H}}
\newcommand{\Fock}{\mathcal{F}}
\newcommand{\Number}{\mathcal{N}}
\newcommand{\vac}{|\Omega\rangle}
\newcommand{\id}{\mathbbm{1}}
\newcommand{\cJ}{\mathcal{J}}
\newcommand{\cS}{\mathcal{S}}
\newcommand{\cL}{\mathcal{L}}
\newcommand{\cO}{\mathcal{O}}

\newcommand{\cQ}{\mathcal{Q}}

\newcommand{\FockB}{\mathbb{B}}

\newcommand{\FockT}{\mathbb{T}}
\newcommand{\FockI}{\mathbb{I}}
\newcommand{\FockJ}{\mathbb{J}}
\newcommand{\FockF}{\mathbb{F}}

\newcommand{\tB}{\widetilde{B}}

\newcommand{\bE}{\mathbb{E}}
\newcommand{\HS}{\mathrm{HS}}

\let\textl\l
\renewcommand{\l}{\ell}
\renewcommand{\i}{\mathrm{i}}
\newcommand{\e}{\mathrm{e}}
\newcommand{\hc}{\mathrm{h.c.}}

\newcommand{\sym}{\mathrm{sym}}
\newcommand{\Tr}{\mathrm{Tr}}

\newcommand{\bPhi}{{\boldsymbol{\phi}}}

\newcommand{\bj}{\boldsymbol{j}}

\newcommand{\btau}{\boldsymbol{\tau}}
\newcommand{\bl}{{\boldsymbol{\l}}}
\newcommand{\bxi}{\boldsymbol{\xi}}

\renewcommand{\hat}[1]{\widehat{#1}}
\renewcommand{\tilde}[1]{\widetilde{#1}}

\newcommand{\lr}[1]{\left\langle #1 \right\rangle}
\newcommand{\norm}[1]{\lVert#1\rVert}
\newcommand{\onorm}[1]{\lVert#1\rVert_\mathrm{op}}

\renewcommand{\d}{\mathop{}\!\mathrm{d}}
\newcommand{\dx}{\d x}
\newcommand{\dy}{\d y}

\newcommand{\ds}{\d s}

\newcommand{\Ubar}{\overline{U}}
\newcommand{\Vbar}{\overline{V}}

\newcommand{\ad}{a^\dagger}
\newcommand{\Ad}{A^\dagger}

\newcommand\mydots{,\makebox[1em][c]{.\hfil.\hfil.},}
\newcommand\mycdots{\makebox[1em][c]{$\cdot$\hfil$\cdot$\hfil$\cdot$}}

\makeatletter
\DeclareFontFamily{OMX}{MnSymbolE}{}
\DeclareSymbolFont{MnLargeSymbols}{OMX}{MnSymbolE}{m}{n}
\SetSymbolFont{MnLargeSymbols}{bold}{OMX}{MnSymbolE}{b}{n}
\DeclareFontShape{OMX}{MnSymbolE}{m}{n}{
    <-6>  MnSymbolE5
   <6-7>  MnSymbolE6
   <7-8>  MnSymbolE7
   <8-9>  MnSymbolE8
   <9-10> MnSymbolE9
  <10-12> MnSymbolE10
  <12->   MnSymbolE12
}{}
\DeclareFontShape{OMX}{MnSymbolE}{b}{n}{
    <-6>  MnSymbolE-Bold5
   <6-7>  MnSymbolE-Bold6
   <7-8>  MnSymbolE-Bold7
   <8-9>  MnSymbolE-Bold8
   <9-10> MnSymbolE-Bold9
  <10-12> MnSymbolE-Bold10
  <12->   MnSymbolE-Bold12
}{}
\let\llangle\@undefined
\let\rrangle\@undefined
\DeclareMathDelimiter{\llangle}{\mathopen}
                     {MnLargeSymbols}{'164}{MnLargeSymbols}{'164}
\DeclareMathDelimiter{\rrangle}{\mathclose}
                     {MnLargeSymbols}{'171}{MnLargeSymbols}{'171}
\makeatother

\newcommand\smallO[1]{
        \mathchoice
            {
                \ensuremath{\mathop{}\mathopen{}{\scriptstyle\mathcal{O}}\mathopen{}\left(#1\right)}
            }
            {
                \ensuremath{\mathop{}\mathopen{}{\scriptstyle\mathcal{O}}\mathopen{}\left(#1\right)}
            }
            {
                \ensuremath{\mathop{}\mathopen{}{\scriptscriptstyle\mathcal{O}}\mathopen{}\left(#1\right)}
            }
            { 
                \ensuremath{\mathop{}\mathopen{}{o}\mathopen{}\left(#1\right)}
            }
    }

\newcommand{\fHN}{{\fH^N}}

\newcommand{\HN}{H_N}

\newcommand{\PsiN}{{\Psi_N}}

\newcommand{\Uz}{U_0}

\newcommand{\Vz}{V_0}
\newcommand{\Vzbar}{\overline{\Vz}}

\newcommand{\hH}{h}
\newcommand{\eH}{e_\mathrm{H}}
\newcommand{\mH}{\mu_\mathrm{H}}
\newcommand{\cEH}{\mathcal{E}_\mathrm{H}}

\newcommand{\UNp}{U_{N,\varphi}}

\newcommand{\FNp}{{\Fock_\perp^{\leq N}}}

\newcommand{\Fp}{{\Fock_\perp}}

\newcommand{\FockH}{\mathbb{H}}
\newcommand{\FockHz}{\FockH_0}

\newcommand{\FockR}{\mathbb{R}}

\newcommand{\Chi}{{\boldsymbol{\chi}}}

\newcommand{\Chiz}{\Chi_0}

\renewcommand{\P}{\mathbb{P}}

\newcommand{\Np}{\Number_\perp}

\newcommand{\boldKz}{\mathbb{K}_0}
\newcommand{\boldKo}{\mathbb{K}_1}
\newcommand{\boldKt}{\mathbb{K}_2}

\newcommand{\BogU}{\mathbb{U}_\BogV}
\newcommand{\BogV}{\mathcal{V}}

\newcommand{\BogUz}{\mathbb{U}_\BogVz}
\newcommand{\BogVz}{{\mathcal{V}_0}}

\newcommand{\PsiNgs}{\Psi_N^\mathrm{gs}}
\newcommand{\PsiNex}{\Psi_N^\mathrm{ex}}

\newcommand{\ENmex}{\mathcal{C}^{(\eta)}_N}
\newcommand{\ENzex}{\mathcal{C}^{(0)}_N}
\newcommand{\ENoex}{\mathcal{C}^{(1)}_N}

\newcommand{\Egs}{\mathcal{E}_N^\mathrm{gs}}
\newcommand{\Eex}{\mathcal{E}_N^\mathrm{ex}}
\newcommand{\ex}{\mathrm{ex}}
\newcommand{\gs}{\mathrm{gs}}

\newcommand{\Var}{\mathrm{Var}}

\newcommand{\cBN}{\mathcal{B}_N}

\newcommand{\ideal}{\mathrm{iid}}

\title{Weak Edgeworth  expansion for the mean-field Bose gas}
\author{Lea Boßmann\thanks{Institute of Science and Technology Austria, Am Campus 1, 3400 Klosterneuburg, Austria and Ludwig-Maximilians-Universität München, Mathematisches Institut, Theresienstr.\ 39, 80333 München, Germany. \texttt{bossmann@math.lmu.de}}
\
and
Sören Petrat\thanks{School of Science, Constructor University, Campus Ring 1, 28759 Bremen, Germany. \texttt{spetrat@constructor.university}}}

\date{\today}
\begin{document}
\maketitle

\begin{abstract}
We consider the ground state and the low-energy excited states of a system of $N$ identical bosons with interactions in the mean-field scaling regime.
For the ground state, we derive a weak Edgeworth expansion for the fluctuations of bounded one-body operators, which yields corrections to a central limit theorem to any order in $1/\sqrt{N}$. For suitable excited states, we show that the limiting distribution is a polynomial times a normal distribution, and that higher order corrections are given by an Edgeworth-type expansion.
\end{abstract}

\section{Introduction}

A quantum mechanical system of $N$ identical bosons is described by a wave function $\Psi$ that is square integrable and symmetric under the exchange of any two particles, i.e.,
\begin{equation}\label{eqn:permutation:symmetry}
\Psi(x_1\mydots x_i\mydots x_j\mydots x_N)=\Psi(x_1\mydots x_j\mydots x_i\mydots x_N)\,,\qquad i,j\in\{1\mydots N\}\,.
\end{equation}
Hence, $\Psi$ is an element of the symmetric subspace $\fH^N_\sym$ of the $N$-body Hilbert space $\fH^N$, where
\begin{equation}
\fH^N:=\fH^{\otimes N}\,,\qquad \fH^N_\sym:=\fH^{\otimes_\sym N}\,,\qquad \fH:=L^2(\R^d)\,,
\end{equation}
for $d\geq 1$ the spatial dimension of the system and where $\otimes_\sym$ denotes the symmetric tensor product.
We study the statistics of measurements described by self-adjoint operators on $\fH^N$.
In particular, we consider one-body operators on $\fH^N$, i.e., operators of the form
\begin{equation}
B_j=\underbrace{\id\otimes\mycdots\otimes\id}_{j-1}\otimes B\otimes \underbrace{\id\otimes\mycdots\otimes\id}_{N-j}
\end{equation}
for bounded self-adjoint operators $B$ on $\fH$.
Since we consider indistinguishable bosons, we study symmetrized operators, i.e., operators of the form $\sum_{j=1}^N B_j$. An example is the number of particles in a bounded volume $V\subset\R^d$, described by the operator
\begin{equation}
\sum_{j=1}^N\chi_V(x_j)\,,
\end{equation}
where $\chi_V$ denotes the characteristic function on $V$. The goal of this article is to better understand the statistics of such operators. 
\medskip

Due to the permutation symmetry \eqref{eqn:permutation:symmetry}, the family of one-body operators  $\{B_j\}_{j=1}^N$ defines a family of identically distributed random variables, whose distribution is determined by the wave function $\Psi$ via the spectral theorem. 
The probability that the corresponding random variable $B_j$ takes values in a set $A\subset\R$ is given by
\begin{equation}
\mathbb{P}_\Psi(B_j\in A)=\lr{\Psi,\chi_A(B_j)\Psi}\,,
\end{equation}
where $\chi_A$ denotes the characteristic function of the set $A$ and where $\lr{\cdot,\cdot}$ denotes the inner product of $\fHN$. Functions of self-adjoint operators are defined via the functional calculus.
Note that the operators  $\sum_j B_j$ are formally the analogue of sample averages, which, in probability theory, are often interpreted as repeated measurements. This interpretation does not apply in our setting: the operator $\sum_j B_j$ does not describe $N$ single-particle measurements on $N$ copies of the system (these measurements would always be independent of each other).

\medskip
If the $N$-body wave function is a product state, i.e., if $\Psi=\varphi^{\otimes N}$ for some $\varphi\in\fH$, the random variables $B_j$ are independent and identically distributed (i.i.d.). Consequently,
$N^{-1}\sum_j B_j$ satisfies the law of large numbers (LLN), and the fluctuations around the expectation value are, in the limit $N\to\infty$, described by the central limit theorem (CLT). Moreover, for large but finite $N$, the fluctuations can be expanded in an asymptotic Edgeworth series, providing higher order corrections to the central limit theorem to any order in $1/\sqrt{N}$ (see Section \ref{subsec:edgeworth:ideal:gas} for a more detailed discussion).\medskip

A factorized wave function $\Psi=\varphi^{\otimes N}$ describes the ground state of an ideal Bose gas, i.e., a system without interactions between the particles.
In this work, we are interested in the situation where the bosons interact weakly with each other.
We consider a system of $N$ bosons in $\R^d$ described by the many-body Hamiltonian
\begin{equation}\label{HN}
\HN = \sum\limits_{j=1}^N(-\Delta_j + V(x_j)) + \frac{1}{N-1}\sum\limits_{1\leq i<j\leq N}v(x_i-x_j)
\end{equation}
acting on $\fH^N_\sym$, under suitable assumptions on the interaction $v$ and the external trapping potential $V$ (see Section \ref{subsec:assumptions}). This describes a Bose gas  in the so-called mean-field (or Hartree) regime, where the interactions are weak and long-ranged.
We consider the ground state $\PsiNgs$ and suitable low-energy excited states $\PsiNex$ of the Hamiltonian $\HN$, i.e.,
\begin{equation}\label{eqn:gs}
\HN\PsiNgs= \mathcal{E}_N^\gs\PsiNgs\,,
\qquad \PsiNgs\in\fH^N_\sym
\end{equation}
and
\begin{equation}
\HN\PsiNex= \mathcal{E}_N^\ex\PsiNex\,,
\qquad \PsiNex\in\fH^N_\sym\,, 
\end{equation}
where $\Egs:=\inf\mathrm{spec}(\HN)$ is the ground state energy and $\Eex$ denotes a suitable excited eigenvalue of $\HN$ (see Definition \ref{def:ENmex}). 
Due to the interactions between the particles, these states are no product states but correlated.
Consequently, the family $\{B_j\}_j$ of one-body operators defines a family of (weakly) dependent random variables. In fact, one can deduce from \cite{spectrum} that their covariance is
\begin{equation}\label{cov}
\text{Cov}_{\PsiN}[B_i,B_j] := \bE_\PsiN[B_i B_j] - \bE_\PsiN[B_i]\, \bE_\PsiN[B_j] = \mathcal{O}(N^{-1})  \qquad (i\neq j),
\end{equation}
where $\bE_\PsiN[\cdot]:=\lr{\PsiN,\cdot \,\PsiN}$.
Despite this dependence, the family $\{B_j\}_j$ satisfies a LLN, which is comparable to the situation of  i.i.d.\ random variables (see Section \ref{subsec:LLN}).
Moreover, one can prove a CLT (see, e.g., \cite{benarous2013, buchholz2014,rademacher2019_2,rademacher2021}), which is a result  of the formal analogy of quasi-free states and Gaussian random variables. Due to the dependence of the random variables $\{B_j\}$, the variance of the limiting Gaussian in the CLT is not given by $\Var_\varphi[B]$ but differs by $\mathcal
{O}(1)$
(see Section \ref{subsec:CLT}).\footnote{Strictly speaking, this implies that the result is no (standard) CLT in the classical sense of probability theory. However, this notion has been used in all previous works in the context of the Bose gas (\cite{benarous2013, buchholz2014,rademacher2019_2,rademacher2021}), and we use it here as well.
}

In this work, we prove that the statistics of bounded one-body operators with respect to the $N$-body ground state $\PsiNgs$ admit a weak Edgeworth expansion, which differs from the expansion for the i.i.d.\ case due to the interactions. Moreover, we prove an Edgeworth-type expansion for a class of low-energy excited states $\PsiNex$.

\subsection{Assumptions}\label{subsec:assumptions}

It is well known that the ground state $\PsiNgs$ as well as the low-energy excited states $\PsiNex$ of $\HN$  exhibit Bose--Einstein condensation (BEC), i.e., 
\begin{equation}\label{RDM}
\lim\limits_{N\to\infty}\Tr_{\fH^k}\left|\gamma^{(k)}_N-|\varphi\rangle\langle\varphi|^{\otimes k}\right|=0
\end{equation}
for any $k\geq 0$. Here, $|\varphi\rangle\langle\varphi|$ denotes the projector onto $\varphi\in\fH$, i.e., the operator with integral kernel $\varphi(x)\overline{\varphi(y)}$, and $\gamma_N^{(k)}$ denotes the $k$-particle reduced density matrix of $\PsiN\in\{\PsiNgs,\PsiNex\}$, whose integral kernel is defined as
\begin{equation}
\gamma_N^{(k)}(x_1\mydots x_k;y_1\mydots y_k):=\int_{\R^{(N-k)d}}\PsiN(x_1\mydots x_N)\overline{\PsiN(y_1\mydots y_k,x_{k+1}\mydots x_N)}\dx_{k+1}\mycdots \dx_N\,.
\end{equation}
The condensate wave function $\varphi$ is given by the minimizer of the Hartree energy functional,
\begin{equation}\label{def:Hartree:functional}
\cEH[\phi]:=\int\limits_{\R^d}\left(|\nabla \phi(x)|^2+V(x)|\phi(x)|^2\right)\dx
+\tfrac12\int\limits_{\R^{2d}}v(x-y)|\phi(x)|^2|\phi(y)|^2\dx\dy\,,
\end{equation}
for $\phi\in\cQ(-\Delta+V)$ under the mass constraint $\norm{\phi}_\fH=1$. 
The minimizer $\varphi$ solves the stationary Hartree equation $h\varphi=0$ in the sense of distributions, where $h$ is the operator on $\D(h)=\D(-\Delta+V)\subset\fH$ defined by
\begin{equation}\label{h:H}
\hH:=-\Delta+V+v*\varphi^2-\mH\,,\qquad
\mH:=\lr{\varphi,\left(-\Delta+V+v*\varphi^2\right)\varphi}\,.
\end{equation} 
The corresponding Hartree energy is denoted by
\begin{equation}
\eH := \cEH[\varphi]\,.
\end{equation}
We make the following assumptions on the interaction potential $v$ and the trap $V$, which, in particular, ensure that $\varphi$ is unique and can be chosen real-valued:

\begin{assumption}\label{ass:V}
Let $V:\R^d\to\R$ be measurable, locally bounded and  non-negative and let $V(x)$ tend to infinity as $|x|\to  \infty$, i.e.,
\begin{equation}
\inf\limits_{|x|>R}V(x)\to\infty \text{ as } R\to \infty\,.
\end{equation}
\end{assumption}
\begin{assumption}\label{ass:v}
Let $v:\R^d\to\R$ be measurable with $v(-x)=v(x)$ and $v\not\equiv 0$, and assume that there exists a constant $C>0$ such that, in the sense of operators on $\cQ(-\Delta)=H^1(\R^d)$,
\begin{equation}
|v|^2\leq  C\left(1-\Delta\right)\,.  \label{eqn:ass:v:2:Delta:bound}
\end{equation}
Besides, assume that $v$ is of positive type, i.e., that it has a non-negative Fourier transform.
\end{assumption}
\begin{assumption}\label{ass:cond}
Assume that there exist constants $C_1\geq0$ and $0<C_2\leq 1$, as well as a function $\varepsilon:\N\to\R_0^+$ with 
\begin{equation}
\lim\limits_{N\to\infty} N^{-\frac13}\varepsilon(N) \leq C_1\,,
\end{equation}
such that 
\begin{equation}\label{eqn:ass:cond}
\HN-N\eH\geq C_2 \sum\limits_{j=1}^N\hH_j-\varepsilon(N)
\end{equation}
in the sense of operators on $\D(\HN)$.
\end{assumption}

Assumption \ref{ass:V} ensures that $V$ is a confining potential; an example is the harmonic oscillator potential, $V(x)=x^2$. 
Assumption~\ref{ass:cond} ensures that low-energy eigenstates of $\HN$ exhibit complete BEC in the Hartree minimizer, with a sufficiently strong rate.
Assumptions \ref{ass:v} and \ref{ass:cond} are, for example, satisfied by any bounded and positive definite interaction potential $v$, and by the repulsive three-dimensional Coulomb potential, $v(x)=1/|x|$. 

Assumptions \ref{ass:V} to \ref{ass:cond} are precisely the assumptions made in \cite{spectrum}. They ensure that we can expand the low-energy eigenstates of $\HN$ and the corresponding energies in an asymptotic series in $1/\sqrt{N}$ (see Section  \ref{subsec:eigenstates:HN}), which is crucial for deriving the Edgeworth expansions.

Our main result holds for the ground state $\PsiNgs$ of $\HN$ and for a class of excited eigenstates $\PsiNex\in\ENmex$. The set $\ENmex\subset \fH^N_\sym$ consists of all eigenstates $\PsiNex$ of $\HN$ where $\HN\PsiNex=\Eex\PsiNex$ such that $\Eex-N\eH$ converges to a non-degenerate eigenvalue of the Bogoliubov Hamiltonian, and where the corresponding Bogoliubov eigenstate is a state with $\eta$ quasi-particles (see Definition \ref{def:ENmex}). In particular, the ground state $\PsiNgs$ is contained in $\ENmex$ for $\eta=0$.

\subsection{Main Result}
We are interested in the statistics of the symmetrized operators $\sum_j B_j$. After centering around the expectation value, we rescale by dividing by $\sqrt{N}$. This scaling is chosen as it is  the size of the standard deviation of $\sum_j B_j$, which follows from  \eqref{cov} and \eqref{RDM} because
\begin{equation}
\Var_{\PsiN}\Bigg[\sum_{j=1}^N B_j\Bigg]
=  \sum_{1\leq j\neq k\leq N}\text{Cov}_{\PsiN}[ B_jB_k]+\sum_{j=1}^N\Var_{\PsiN}[B_j]
=\mathcal{O}(N)\,.
\end{equation}
This leads to the random variable
\begin{equation}
\cBN:=\frac{1}{\sqrt{N}}\sum_{j=1}^N(B_j-\bE_\PsiN[B])
\end{equation}
for self-adjoint $B\in\cL(\fH)$, 
where $\bE_\PsiN$ denotes the expectation value of a random variable  with respect to the probability distribution determined by $\PsiN$. 
Moreover, we consider operators $B$ such that the Hartree minimizer $\varphi$ is not an eigenstate of $B$. This is equivalent to the statement that the standard deviation $\sigma$ of the limiting Gaussian in the CLT (see our theorem below) is nonzero, see \eqref{eqn:sigma}.
Our main result is the following:

\begin{theorem}\label{thm:edgeworth}
Let Assumptions \ref{ass:V} to \ref{ass:cond} hold and let $\PsiN\in\ENmex$ for some $\eta\in\N_0$, with $\ENmex$ as in Definition \ref{def:ENmex}. Let $a\in \N_0$ and $g\in L^1(\R)$ such that its Fourier transform $\hat{g}\in L^1(\R,(1+|s|^{3a+4})$.
Then, for any self-adjoint bounded operator $B\in\cL(\fH)$ such that the Hartree minimizer $\varphi$ is not an eigenstate of $B$,
\begin{equation}\label{eqn:thm}
\left|\mathbb{E}_{\PsiN}[g(\cBN)]
- \sum_{j=0}^aN^{-\frac{j}{2}}
\int\dx\, g(x)\mathfrak{p}_j(x) \, \frac{1}{\sqrt{2\pi\sigma^2}} \e^{-\frac{x^2}{2\sigma^2}} \right|
\leq C_B(a,g) N^{-\frac{a+1}{2}}\,
\end{equation}
for $\sigma$ as in \eqref{eqn:sigma}. Here, the functions $\mathfrak{p}_j(x)$ are polynomials of finite degree with real coefficients depending on $B$, $V$ and $v$. The error can be estimated as
\begin{equation}\label{eqn:C_B(a,g)}
C_B(a,g)\leq C(a) \big( 1+\onorm{B}^{3a+4} \big) \int_\R\ds\, |\hat{g}(s)| \left(1+|s|^{3a+3}+N^{-\frac{1}{2}}|s|^{3a+4}\right)
\end{equation}
for some $C(a)>0$, where $\onorm{\cdot}$ denotes the operator norm on $\cL(\fH)$.
\thmit{
\item \label{thm:gs}
If $\PsiN=\PsiNgs\in\ENzex$, then  $\mathfrak{p}^\gs_j$ is a polynomial of degree $3j$ which is even/odd for $j$ even/odd. In particular,
\begin{subequations}
\begin{eqnarray}
\mathfrak{p}^\gs_0 (x)&=&1\,,\\
\mathfrak{p}^\gs_1(x) &=& \frac{\alpha_3}{6\sigma^3}H_3\left(\frac{x}{\sigma}\right)\,,
\end{eqnarray}
\end{subequations}
with $\alpha_3$ as in \eqref{def:alpha} and where $H_3$ is the third Hermite polynomial (see \eqref{def:Hermite}).

\item \label{thm:ex}
If $\PsiN=\PsiNex\in \ENmex$ for some $\eta>0$, then $\mathfrak{p}_j^\ex$ is a polynomial of degree $3j+2\eta$ which is even/odd for $j$ even/odd. The leading order $\mathfrak{p}_0^\ex$ is computed in Proposition~\ref{thm:no:clt}.
}
\end{theorem}

\begin{remark}\label{rem:CLT}
Theorem \ref{thm:edgeworth} implies a quantitative version of the CLT for the ground state with improved rate. 
Following the proof of \cite[Corollary 1.2]{buchholz2014}, we approximate the characteristic function $\chi_{[\alpha,\beta]}$ for some $\alpha,\beta\in\R$ from below and above by some smooth and compactly supported functions $g^\varepsilon_-$ and $g^\varepsilon_+$. For $\varepsilon>0$, we define these functions as
$g^\varepsilon_{\pm}:=\chi_{[\alpha\mp\varepsilon,\beta\pm\varepsilon]}*\zeta_\varepsilon$
for $\zeta_\varepsilon(x)=\varepsilon^{-1}\zeta(x/\varepsilon)$, where $\zeta\in\mathcal{C}^\infty_c(\R)$ is some non-negative function such that $\zeta(x)=0$ for $|x|>1$ and $\int_\R \zeta=1$. 
Consequently, 
\begin{equation}
\bE_{\PsiNgs}[g^\varepsilon_-(\cBN)]\leq\P_{\PsiNgs}(\cBN\in[\alpha,\beta])\leq \bE_{\PsiNgs}[g^\varepsilon_+(\cBN)]\,.
\end{equation}
Analogously to \cite{buchholz2014}, one obtains the estimate $|\hat{g}^\varepsilon_{\pm}(s)|\leq C|\hat{\zeta}(\varepsilon s)| \min\{|s|^{-1},|\beta-\alpha|\}$ for some constant $C>0$, hence Theorem \ref{thm:edgeworth} leads (for any $a\in\N_0$)  to
\begin{equation}
\left|\bE_{\PsiNgs}\left[g^\varepsilon_{\pm}(\cBN)\right]-\int g^\varepsilon_{\pm}(x)b_a(x)\dx\right|
\leq C(a\left(N^{-\frac{a+1}{2}}\varepsilon^{-(3a+3)}+N^{-\frac{a+2}{2}}\varepsilon^{-(3a+4)}\right)\,,
\end{equation}
where the constant $C$ depends on $B$, $\alpha$ and $\beta$ and where we abbreviated
\begin{equation}
b_a(x):=\sum_{j=0}^a N^{-\frac{j}{2}}\mathfrak{p}_j^\gs(x)\frac{1}{\sqrt{2\pi\sigma^2}}\e^{-\frac{x^2}{2\sigma^2}}\,.
\end{equation}
Since $|\int_\R g^\varepsilon_{\pm}b_a -\int_\alpha^\beta b_a(x)|\leq C\varepsilon$, this yields
\begin{equation}\label{eqn:remark:CLT}
\left|\P_{\PsiNgs}(\cBN\in[\alpha,\beta])-\int_\alpha^\beta b_a(x)\dx\right|\leq C(a)\left(\varepsilon+N^{-\frac{a+1}{2}}\varepsilon^{-(3a+3)}+N^{-\frac{a+2}{2}}\varepsilon^{-(3a+4)}\right)\,.
\end{equation}
The right hand side of \eqref{eqn:remark:CLT} is minimal for  $\varepsilon=N^{-\frac{a+1}{6a+8}}$, which, in particular, implies that it is always larger than $N^{-\frac16}$.
Consequently, choosing $a$ sufficiently large yields
\begin{equation}
\left|\P_{\PsiNgs}(\cBN\in[\alpha,\beta])-\frac{1}{\sqrt{2\pi\sigma^2}}\int_\alpha^\beta \e^{-\frac{x^2}{2\sigma^2}}\dx\right|\leq C_\gamma N^{-\gamma}\qquad\text{for any }\gamma<\frac16\,.
\end{equation}
This improves the previous estimate $N^{-1/8}$, which follows analogously to \cite{rademacher2019_2} by taking into account only the leading order $a=0$.
\end{remark}

\begin{remark}\label{rem:weak:EE}
Theorem \ref{thm:edgeworth} constitutes a weak Edgeworth expansion as introduced in \cite{goetze1983,breuillard2005, fernando2021}. In particular, our result does not imply an asymptotic expansion of the probability $\P_{\PsiN}(\cBN\in[\alpha,\beta])$.
The reason why we can only state our result in this weak form is that our error estimate when truncating the expansion of the characteristic function $\bE_{\PsiN}[\e^{\i s\cBN}]$ grows polynomially in $s$ (Proposition \ref{prop:S}). Hence, we can not simply apply the Fourier transform to obtain an expansion of the probability density. It is an open question whether a strong Edgeworth expansion exists, i.e., whether there exist constants $C_a$ such that
\begin{equation}
\left|\P_{\PsiNgs}\left(\cBN\in[\alpha,\beta]\right)-  \int\limits_\alpha^\beta
\sum\limits_{j=0}^aN^{-\frac{j}{2}}\frac{\mathfrak{p}_j(x)}{\sqrt{2\pi}\sigma}\e^{-\frac{x^2}{2\sigma^2}}\dx\right|\overset{\text{(?)}}{\leq} C_a N^{-\frac{a+1}{2}}\,.
\end{equation}

\end{remark}

If the $N$-body system is in its ground state $\PsiNgs$, Theorem \ref{thm:edgeworth} implies that  $\cBN$ admits a weak Edgeworth expansion although the random variables are not independent. However, the interactions affect the precise form of the Edgeworth series: the standard deviation $\sigma$ of the Gaussian as well as the polynomials $\mathfrak{p}_j^\gs$ differ from the expansion for the non-interacting Bose gas (see Sections \ref{subsec:edgeworth:ideal:gas} and \ref{sec_Edgeworth_int} for a detailed discussion). 
To prove Theorem~\ref{thm:edgeworth}, we expand the characteristic function 
$$\phi_N^\gs(s):=\lr{\PsiNgs,\e^{\i s\cBN}\PsiNgs}$$
in  powers of $N^{-1/2}$. To leading order, $\phi_N^\gs(s)$ is  given by the expectation value of a Weyl operator with respect to a quasi-free state . Quasi-free states satisfy a Wick rule  comparable to Wick's probability theorem for Gaussian random variables, and this formal analogy is the reason why we obtain a CLT for the ground state. Technically, we use an equivalent formulation of Wick's rule, namely the fact that a quasi-free state is a Bogoliubov transformation of the vacuum. This allows us to reduce the computation of $\phi_N^\gs(s)$ to the computation of  vacuum expectation values, which are non-zero only if they contain equal numbers of creation and annihilation operators.

For low-energy excited states, the leading order of the corresponding characteristic function $\phi_N^\ex(s)$ is no longer given by an expectation value with respect to a quasi-free state, but rather a state with a finite number of creation/annihilation operators acting on a quasi-free state. Consequently, the limiting distribution is not a Gaussian but a Gaussian 
multiplied with a polynomial. One still obtains an Edgeworth-type expansion, but each order of the distribution is now the Gaussian times a (different) polynomial.
\medskip

Theorem \ref{thm:edgeworth} is, to the best of our knowledge, the first derivation of an Edgeworth expansion for an interacting quantum many-body system. 
Asymptotic expansions for (weakly) dependent random variables have been derived in \cite{nagaev1957,nagaev1961,herve2010} for Markov processes,
in \cite{goetze1983} for stochastic processes which are approximated by a suitable Markov process, and in \cite{coelho1990} in the context of dynamical systems. In \cite{fernando2021}, the authors prove the existence of Edgeworth expansions for weakly dependent random variables under fairly generic conditions, which includes random variables arising from dynamical systems and Markov chains but excludes our model\footnote{In \cite{fernando2021}, the authors consider a Banach space $B$ and assume that the characteristic function is of the form $\phi_N(s)=\l(\mathcal{L}^N_sv)$, where $\cL_s:B\to B$ is a family of bounded linear operators and where $v\in B$, $\l\in B'$. Applied to our setting, we would identify $v$ with the ground state $\Psi_N$, and $\ell$ with the projection on the ground state. However, $e^{is\cBN}$ is not of the form $\cL_s^N$ for some $N$-independent $\cL_s$. Even if we would introduce $\cL_s = e^{is\frac{1}{N}\cBN}$, this operator would not satisfy the assumptions made in \cite{fernando2021}, which include that the spectrum of $\cL_s$ is contained in the open disc of radius $1$ for all $s \neq 0$, and that $\| \cL_s^N \| \leq \frac{1}{N^{r_2}}$ for some $r_2>0$.}.

As discussed in Section \ref{subsec:edgeworth:ideal:gas} for the i.i.d.\ situation, Theorem \ref{thm:edgeworth} yields a very precise description in the center of the distribution. In contrast, it does not generally provide a good approximation of the tails of the distribution. For the dynamics generated by $\HN$,  large deviation estimates have been proven in \cite{kirkpatrick2021,rademacher2021_2}.

We expect that Theorem \ref{thm:edgeworth} can be generalized to all situations where the $N$-body wave function admits an (explicitly known) asymptotic expansion in the spirit of Lemma \ref{lem:known}. For example, it seems obvious that a dynamical Edgeworth expansion should exist, which provides corrections to \cite{benarous2013}; moreover, generalizations to $k$-body operators as in \cite{rademacher2021} and to $k$ one-body operators as in \cite{buchholz2014} seem feasible.
\medskip

The remainder of the article is structured as follows: 
In Section \ref{sec:framework}, we summarize the quantum many-body framework  and collect known results for the mean-field Bose gas which we require for the proof.
Section \ref{sec:prob} is a review of the probabilistic picture, including existing results on the CLT for the interacting Bose gas. 
In particular,  we analyse the effect of the interactions on the Edgeworth series (Sections \ref{subsec:edgeworth:ideal:gas} and \ref{sec_Edgeworth_int}).
Finally, Section \ref{sec:proofs} contains the proof of  Theorem~\ref{thm:edgeworth}.

\section{Many-body framework}\label{sec:framework}

\subsection{Excitations from the condensate}\label{subsec:pre:excitations}

We consider $N$-body states $\Psi$ which exhibit complete BEC in the Hartree minimizer $\varphi$ in the sense of \eqref{RDM}. However, this does  in general not imply that $\Psi=\varphi^{\otimes N}$; in fact, an $\mathcal{O}(1)$ fraction of the particles forms excitations from the condensate. To describe them mathematically, one recalls, e.g.\  from \cite{lewin2015_2}, that any $\Psi\in\fH^N_\sym$ can be decomposed as
\begin{equation}
\Psi=\sum\limits_{k=0}^N{\varphi}^{\otimes (N-k)}\otimes_\sym\chi^{(k)}\,,
\qquad \chi^{(k)}\in \bigotimes\limits_\sym^k \left\{\varphi\right\}^\perp\,,
\end{equation}
with the usual notation $\left\{\varphi\right\}^\perp :=\left\{\phi\in \fH: \lr{\phi,\varphi}=0\right\}$. 
The sequence
\begin{equation}\label{def:ChiN:Chi}
\Chi:=\big(\chi^{(k)}\big)_{k=0}^N
\end{equation}
of $k$-particle excitations forms a vector in the (truncated) excitation Fock space over $\left\{\varphi\right\}^\perp $,
\begin{equation}\label{Fock:space}
\FNp=\bigoplus_{k=0}^N\bigotimes_\sym^k \left\{\varphi\right\}^\perp 
\;\subset \;
\Fp=\bigoplus_{k=0}^\infty\bigotimes_\sym^k \left\{\varphi\right\}^\perp \,,
\end{equation}
and vectors in $\Fp$ (resp.\ $\FNp$) are denoted as $\bPhi$ (resp.\ $\bPhi_{\leq N}$).
The creation and annihilation operators on $\Fp$, $\ad(f)$ and $a(f)$ for $f\in\left\{\varphi\right\}^\perp$, are defined in the usual way as
\begin{subequations}
\begin{align}
(\ad(f)\bPhi)^{(k)}(x_1\mydots x_k) &= \frac{1}{\sqrt{k}}\sum\limits_{j=1}^kf(x_j)\phi^{(k-1)}(x_1\mydots x_{j-1},x_{j+1}\mydots x_k)\,, \\
(a(f)\bPhi)^{(k)}(x_1\mydots x_k) &= \sqrt{k+1}\int\d x\overline{f(x)}\phi^{(k+1)}(x_1\mydots x_k,x)
\end{align} 
\end{subequations}
for $k\geq 1$ and $k\geq 0$, respectively, and $\bPhi\in\Fp$. They can be expressed in terms of the operator-valued distributions $\ad_x$ and $a_x$,
\begin{equation}
\ad(f)=\int\d x f(x)\,\ad_x\,,\qquad a(f)=\int\d x\overline{f(x)}\,a_x\,,
\end{equation}
which satisfy the canonical commutation relations 
\begin{equation}
[a_x,\ad_y]=\delta(x-y)\,,\qquad [a_x,a_y]=[\ad_x,\ad_y]=0\,.
\end{equation}
We denote the second quantization in $\Fp$  (resp.\ $\Fock$) of an operator $A$  by $\d\Gamma_\perp(A)$ (resp.\ $\d\Gamma(A)$). The vacuum is denoted by $\vac$ and the number operator on $\Fp$ is given by
\begin{equation}
\Np:=\d\Gamma_\perp(\id)\,,\qquad (\Np\bPhi)^{(k)}=k\phi^{(k)}\;\text{ for }
\bPhi\in\Fp\,.
\end{equation}
An $N$-body state $\Psi$ is mapped onto its corresponding excitation vector $\Chi$  by the unitary excitation map  $\UNp$
\begin{equation}\label{intro:UNp}
\UNp:\fHN\to\FNp\,,\qquad \Psi\mapsto \UNp\Psi=\Chi\,,
\end{equation}
introduced in \cite{lewin2015_2}. 
For $f,g\in\{\varphi\}^\perp$, it acts as
\begin{subequations}\label{eqn:substitution:rules}
\begin{eqnarray}
\UNp\, \ad({\varphi})a({\varphi})\UNp^*&=&N-\Np\,,\\
\UNp\, \ad(f)a({\varphi}) \UNp^*&=&\ad(f)\sqrt{N-\Np}\,,\\
\UNp\, \ad({\varphi})a(g)\UNp^*&=&\sqrt{N-\Np}a(g)\,,\\
\UNp\, \ad(f)a(g)\UNp^*&=&\ad(f)a(g)
\end{eqnarray}
\end{subequations}
as identities on $\FNp$. We extend $\UNp$ trivially to a map to the full space $\Fp$.
Analogously, elements of $\FNp$ are naturally understood as elements of $\Fp$.

\subsection{Bogoliubov theory}\label{subsec:framework:Bog}

It was shown in \cite{spectrum} that the low-energy eigenstates of $\HN$ can be retrieved by perturbation theory around the eigenstates of the Bogoliubov Hamiltonian, which is given by
\begin{equation}\label{Bog_Ham_def}
\FockHz:=\boldKz+\boldKo+\boldKt+\boldKt^*\,.
\end{equation}
Here,
\begin{subequations}\label{eqn:K:notation}
\begin{eqnarray}
 \boldKz&:=&\int\dx\,\ad_x  \hH_x a_x\,,\label{eqn:K:notation:0}\\
 \boldKo&:=&\int\dx_1\dx_2\, (qKq)(x_1;x_2)\ad_{x_1} a_{x_2}\,,\label{eqn:K:notation:1}\\
 \boldKt&:=&\tfrac12\int\dx_1\dx_2\, (q_1q_2K)(x_1,x_2)\ad_{x_1}\ad_{x_2}\,,\label{eqn:K:notation:2}
\end{eqnarray}
\end{subequations}
for $h$ from \eqref{h:H}, where $K$  is the operator with kernel
\begin{equation}
K(x;y)=v(x-y)\varphi(x)\varphi(y)
\end{equation}
and where we used the orthogonal projectors
\begin{equation}\label{p:q}
p:=|\varphi\rangle\langle\varphi|\,,\qquad q:=\mathbbm{1}-p
\end{equation}
onto the condensate and its complement. \\

\subsubsection{Bogoliubov transformations}
The Bogoliubov Hamiltonian $\FockHz$ can be diagonalized by Bogoliubov transformations (see, e.g., \cite{solovej_lec}), which are defined as follows:
For $ F=f\oplus \overline{g} \in\{\varphi\}^\perp\oplus\{\varphi\}^\perp$, one defines the generalized creation and annihilation operators $A(F)$ and $\Ad(F)$ as
\begin{equation}\label{eqn:A(F)}
A(F)=a(f)+\ad(g)\,, \quad \Ad(F)=A(\cJ F)=\ad(f)+a(g)\,,
\end{equation}
where $ \cJ=\left(\begin{smallmatrix}0 & J\\J&0\end{smallmatrix}\right)$
with $(Jf)(x)=\overline{f(x)}$. 
An operator $\BogV$ on $\{\varphi\}^\perp\oplus\{\varphi\}^\perp$  such that 
\begin{equation}
\Ad(\BogV F)=A(\BogV\mathcal{J}F)\,,\qquad [A(\BogV F_1),\Ad(\BogV F_2)]=[A(F_1),\Ad(F_2)]\,,
\end{equation}
is called a (bosonic) Bogoliubov map. It can be written in the block form
\begin{equation}\label{BogV:block:form}
\BogV:=\begin{pmatrix}U & \Vbar\\V & \Ubar\end{pmatrix}\,,\quad U,V:\{\varphi\}^\perp\to \{\varphi\}^\perp\,,
\end{equation}
where $\overline{U}$ and $\overline{V}$ denote the operators with integral kernels $\overline{U(x,y)}$ and $\overline{V(x,y)}$, respectively.
If $V$ is Hilbert-Schmidt, $\BogV$ is unitarily implementable on $\Fp$, i.e., there exists a unitary transformation $\BogU:\Fp\to\Fp$, called a Bogoliubov transformation, such that
\begin{equation}\label{eqn:unit:impl}
\BogU A(F)\BogU^*=A(\BogV F)\,.
\end{equation}
The identity \eqref{eqn:A(F)} leads to a transformation rule of creation/annihilation operators under a Bogoliubov transformation,
\begin{equation}\begin{split}
\BogU \,a(f)\,\BogU^*&\;=\;a(Uf) +\ad(\overline{Vf})\,,\\
\BogU\,\ad(f)\,\BogU^*&\;=\;a(\overline{Vf})+\ad(Uf)\label{eqn:trafo:ax}
\end{split}\end{equation}
for $f\in\{\varphi\}^\perp$.
In particular, powers of $\Np$ conjugated with $\BogU$ can be bound as 
\begin{equation}\label{eqn:Number:BT}
\BogU(\Np+1)^b\BogU^* \leq C_\BogV^b\, b^b(\Np+1)^b \qquad (b\in\N)
\end{equation}
in the sense of operators on $\Fp$, where $C_\BogV:=2\norm{V}_\HS^2+\onorm{U}^2+1$ \cite[Lemma 4.4]{QF}.

\subsubsection{Quasi-free states}\label{subsubsec:QF}
Finally, we recall that a normalized state $\bPhi\in\Fp$ is called quasi-free if there exists a Bogoliubov transformation $\BogU$ such that
\begin{equation}
\bPhi=\BogU|\Omega\rangle\,.
\end{equation}
Quasi-free states satisfy Wick's rule (e.g. \cite[Theorem 1.6]{nam2011}: for $\Phi$ quasi-free, it holds that
\begin{subequations}\label{eqn:Wick}
\begin{align}
\lr{\bPhi,a^\sharp(f_1)\mycdots a^\sharp(f_{2n-1})\bPhi}_{\Fp} &= 0\,,\\
\lr{\bPhi,a^\sharp(f_1)\mycdots a^\sharp(f_{2n})\bPhi}_{\Fp} &= \sum\limits_{\sigma\in P_{2n}}\prod\limits_{j=1}^n \lr{\bPhi,a^\sharp(f_{\sigma(2j-1)})a^\sharp(f_{\sigma(2j)})\bPhi}_{\Fp}
\end{align}
\end{subequations}
for $a^\sharp\in\{\ad,a\}$, $n\in\N$ and $f_1\mydots f_{2n}\in \{\varphi\}^\perp$.
Here, $P_{2n}$ denotes the set of pairings
\begin{equation}
P_{2n}:=\{\sigma\in\mathfrak{S}_{2n}:\sigma(2a-1)<\min\{\sigma(2a),\sigma(2a+1)\} \;\forall a\in \{1,2\mydots 2n\} \}\,,
\end{equation}
where $\mathfrak{S}_{2n}$ denotes the symmetric group on the set $\{1,2\mydots 2n\}$.

\subsubsection{Eigenstates of $\FockH_0$}
We denote by $\BogUz:\Fp\to\Fp$ the Bogoliubov transformation that diagonalizes $\FockHz$, i.e.,
\begin{equation}\label{eqn:BT:diag}
\BogUz\FockHz\BogUz^*=\d\Gamma_\perp(D)+\inf\sigma(\FockHz)\,,
\end{equation} 
where $D>0$ is a self-adjoint operator on $\{\varphi\}^\perp$. It admits a complete set of normalized eigenfunctions, denoted as $\{\xi_j\}_{j\geq0}$.
The ground state $\Chi_0^\gs$ of $\FockHz$ is unique and given by
\begin{equation}
\Chiz^\gs=\BogUz^*\vac\,.
\end{equation}
Any non-degenerate excited eigenstate $\Chiz^\ex$ of $\FockHz$ can be expressed as
\begin{equation}\label{eqn:excited:states:FockHz}
\Chiz^\ex=\BogUz^*
\frac{\big(\ad(\xi_0)\big)^{\eta_0}}{\sqrt{\eta_0!}}
\frac{\big(\ad(\xi_1)\big)^{\eta_1}}{\sqrt{\eta_1!}}
\,\mycdots\,
\frac{\big(\ad(\xi_k)\big)^{\eta_{k}}}{\sqrt{\eta_{k}!}}
\vac
\end{equation}
for some $k\in\N_0$ and some tuple $(\eta_0\mydots \eta_{k})\in\N_0^{k+1}$.
Finally, the Bogoliubov map corresponding to $\BogUz$ is denoted by
\begin{equation}\label{eqn:U_0:V_0}
\BogVz=\begin{pmatrix} U_0 & \Vbar_0 \\ V_0 & \Ubar_0 \end{pmatrix}\,.
\end{equation}

\subsection{Low-energy eigenstates of $\HN$}\label{subsec:eigenstates:HN}

Assumptions \ref{ass:V} to \ref{ass:cond} ensure that $\HN$ has a unique ground state and a discrete low-energy spectrum. We will consider the following class of eigenstates of $\HN$:
\begin{definition}\label{def:ENmex}
Let $\eta\in\N_0$. Then $\PsiN\in\fH^N_\sym$ is an element of the set $\ENmex$ iff all of the following are satisfied:
\defit{
\item $\PsiN$ is an eigenstate of $\HN$, i.e., $\HN\PsiN=\mathcal{E}_N\PsiN$.
\item There exists a non-degenerate Bogoliubov eigenstate, $\FockH_0\Chi_0=E_0\Chi_0$, such that 
$$\lim\limits_{N\to\infty}\left(\mathcal{E}_N-N\eH\right)=E_0\,.$$
\item $\Chi_0$ is a state with $\eta$ quasi-particles, i.e., it is given by \eqref{eqn:excited:states:FockHz} with 
$\eta_0+\eta_1+\dots+\eta_k=\eta$.
}
In particular,
\begin{equation}
\PsiNgs\in\ENzex\,, 
\end{equation}
i.e., the ground state is contained in the set $\ENmex$ with zero quasi-particles.\end{definition}

To keep the notation simple, we will indicate the quasi-particle number $\eta$ only when it is inevitable to avoid ambiguities.
If $\PsiN\in\ENmex$ for some $\eta\in\N_0$, it was shown in  \cite[Theorem~3]{spectrum} that $\Chi=\UNp\Psi$ admits an asymptotic expansion in the parameter $(N-1)^{-1/2}$, namely
\begin{equation}\label{eqn:expansion:Chin:tilde}
\Big\|{\Chi-\sum\limits_{\l=0}^a(N-1)^{-\frac{\l}{2}}\tilde{\Chi}_\l}\Big\|\leq C(a)(N-1)^{-\frac{a+1}{2}}
\end{equation}
for some constant $C(a)>0$ and for coefficients $\tilde{\Chi}_\l\in\Fp$ given in \cite[Theorem 3, Eqn.\ (3.26)]{spectrum}.

For the proof of Theorem \ref{thm:edgeworth}, it is more convenient to have a full expansion of these states in powers of $N^{-1/2}$ instead of $(N-1)^{-1/2}$, which can be deduced from the results in \cite{spectrum} in a straightforward way.

\begin{lem}\label{lem:known}
Let Assumptions \ref{ass:V} to \ref{ass:cond} hold, let $\PsiN\in\ENmex$ for some $\eta\in\N_0$ and denote the corresponding excitation vector by $\Chi=\UNp\Psi$.
\lemit{
\item \label{lem:known:expansion:Chi:eta}
For any $a\in\N_0$, there exists a constant $C(a)>0$ such that
\begin{equation}\label{eqn:expansion:Chi:spectrum:eta}
\Big\|{\Chi-\sum\limits_{\l=0}^aN^{-\frac{\l}{2}}\Chi_\l}\Big\|\leq C(a)N^{-\frac{a+1}{2}}\,,
\end{equation}
where
\begin{equation}\label{eqn:def:Theta:spectrum:eta}
\Chi_\l=\BogUz^*\sum_{\substack{j=0\\\l+\eta+j\text{ even}}}^{3\l+\eta}\int\dx^{(j)}\Theta^{(\eta)}_{\l,j}(x^{(j)})\ad_{x_1}\mycdots\ad_{x_j}|\Omega\rangle
\end{equation}
for some functions $\Theta^{(\eta)}_{\l,j}\in L^2_\sym(\R^{dj})$.

\item \label{lem:known:moments:Chil}
For any $\l,b\in\N$, there exists a constant $C(\l,b)$ such that
\begin{equation}\label{lem:moments:Chil}
\norm{(\Np+1)^b\Chi_\l}\leq C(\l,b)\,.
\end{equation}

\item \label{lem:known:expansion:B}
Let $B\in\cL(\fH)$. For any $a\in\N_0$, there exists some constant $C(a)>0$ such that
\begin{equation}
\left|\lr{\PsiN,B_1\PsiN}-\sum_{\l=0}^a N^{-\l}B^{(\l)}\right|\leq C(a)\onorm{B}N^{-(a+1)}\,,
\end{equation}
where the coefficients 
\begin{equation}\label{eqn:B^(l)}
B^{(\l)}:=\sum_{k=1}^\l\tbinom{\l-1}{\l-k}\Tr\gamma_{1,k}B\in\R
\end{equation}
can be bounded as
\begin{equation}
|B^{(\l)}|\leq C(\l)\onorm{B}
\end{equation}
for some constants $C(\l)>0$. In particular, $B^{(0)}=\lr{\varphi,B\varphi}$, and
\begin{equation}\begin{split}
B^{(1)} = &\lr{\Chiz,\left(\ad(qB\varphi)+a(qB\varphi)\right) \Chi_1} 
+\lr{\Chi_1, \left(\ad(qB\varphi)+a(qB\varphi)\right) \Chiz}\\ 
&+ \lr{\Chiz,\d\Gamma(q \tB q) \Chiz}.
\end{split}\end{equation}

}
\end{lem}

The functions $\Theta^{(\eta)}_{\ell,j}$ can be computed using perturbation theory, and we refer to \cite{spectrum} for the explicit expressions. In a similar way, one obtains explicit expressions for $B^{(\l)}$; see~\cite{proceedings}.

\section{Probabilistic picture}\label{sec:prob}
To illustrate the effect of the interactions, we compare in this section the random variables with probability distribution determined by $\PsiN\in\ENmex$ (for some $\eta\in\N_0$) with the random variables distributed according to the product state
\begin{equation}
\Psi_N^\ideal:=\varphi^{\otimes N}\,.
\end{equation} 
To underline differencences between the ground state 
$\PsiNgs\in\ENzex$ and excited states $\PsiNex\in\ENmex$  for $\eta>0$, we will indicate this in the notation by using the superscripts $^\gs$ and $^\ex$ when appropriate.

\subsection{Random variables}

A self-adjoint one-body operator $B\in\cL(\fH)$ defines a family $\{B_j\}_{j=1}^N$ of random variables with common probability distribution determined by the $N$-body wave function $\PsiN$. For $\Psi_N^\ideal$, the random variables are i.i.d., and the expectation value $\mathbb{E}_{\varphi}[B]$, the variance $\Var_{\varphi}[B]$ and the standard deviation $\sigma_\ideal$ are given by 
\begin{equation}
\mathbb{E}_{\varphi}[B]=\lr{\varphi,B\varphi}\,,\qquad
\Var_{\varphi}[B]=\sigma_\ideal^2=\lr{\varphi,B^2\varphi}-\lr{\varphi,B\varphi}^2\,.
\end{equation} 
For an eigenstate $\PsiN\in\ENmex$ of $\HN$, the random variables are no longer independent, and the corresponding quantities $\bE_\PsiN[B]$, $\Var_\PsiN[B]$ and $\sigma_N$  can be computed as
\begin{eqnarray}
\mathbb{E}_{\PsiN}[B]&=&\frac{1}{N}\sum_{j=1}^N\lr{\PsiN,B_j\PsiN}=\lr{\PsiN,B_1\PsiN}\,,\\
\Var_\PsiN[B]&=&\sigma_N^2\;=\;\lr{\PsiN,B_1^2\PsiN}-\lr{\PsiN,B_1\PsiN}^2
\end{eqnarray}
due to the bosonic symmetry \eqref{eqn:permutation:symmetry} of $\PsiN$.
Note that by \eqref{RDM},
\begin{equation}\label{eqn:limits:E:Var}
\lim\limits_{N\to\infty}\bE_{\PsiN}[B]=\bE_{\varphi}[B]\,,\qquad
\lim\limits_{N\to\infty}\Var_{\PsiN}[B]=\Var_{\varphi}[B]\,.
\end{equation}

\subsection{Law of large numbers}\label{subsec:LLN}
For the product state $\Psi_N^\ideal$, the weak LLN states that the empiric mean converges to its expectation value, i.e., 
\begin{equation}
\lim\limits_{N\to\infty}\P_{\Psi_N^\ideal} \left(\left|\frac{1}{N}\sum_{j=1}^N B_j-\lr{\varphi,B\varphi}\right|\geq\varepsilon\right)=0
\end{equation}
for any $\varepsilon>0$. Abbreviating $\tilde{B}:=B-\lr{\varphi,B\varphi}$, Markov's inequality yields for the interacting gas  (see, e.g., \cite[Sec.\ 1]{benarous2013})
\begin{eqnarray}
\P_{\PsiN} \left(\left|\frac{1}{N}\sum_{j=1}^N \tilde{B}_j\right|\geq\varepsilon\right)
&\leq&\frac{1}{N^2\varepsilon^2}\lr{\PsiN,\Big(\sum_{j=1}^N \tilde{B}_j\Big)^2\PsiN}\nonumber\\
&\leq&\varepsilon^{-2}\lr{\PsiN,\tilde{B}_1\tilde{B}_2\PsiN}+N^{-1}\varepsilon^{-2}\lr{\PsiN,\tilde{B}_1^2\PsiN}\,,
\end{eqnarray}
hence \eqref{RDM} yields
\begin{equation}\label{eqn:LLN}
\lim\limits_{N\to\infty}\P_{\PsiN} \left(\left|\frac{1}{N}\sum_{j=1}^N B_j-\lr{\varphi,B\varphi}\right|\geq\varepsilon\right)=0\,.
\end{equation}
The LLN for $\PsiN$ looks formally like the LLN for independent random variables. Let us stress that $\Psi_N^\ideal$ is not the ground state of the ideal gas because $\varphi$ is the minimizer of the Hartree energy functional, which depends on the interactions. In this sense, the interactions have an effect already on the level of the LLN.

\subsection{Central limit theorem for the ground state}\label{subsec:CLT}
Let us first compare the ground state $\PsiNgs$ of the interacting gas with the product state $\Psi_N^\ideal$.
The fluctuations around the respective expectation values are described by the rescaled and centred random variables
\begin{equation}\label{cBN}
\cBN^\ideal:=\frac{1}{\sqrt{N}}\sum_{j=1}^N\left(B_j-\mathbb{E}_{\varphi}[B]\right)\,,\qquad
\cBN=\frac{1}{\sqrt{N}}\sum_{j=1}^N\left(B_j-\mathbb{E}_{\PsiN}[B]\right)\,.
\end{equation}
For the i.i.d.\ situation, the CLT states that the distribution of $\cBN^\ideal$ converges to the centred Gaussian distribution with variance $\sigma_\ideal^2$, i.e.,
\begin{equation}\label{eqn:CLT:ideal}
\lim\limits_{N\to\infty}\left|\P_{\Psi_N^\ideal}(\cBN^\ideal\in A)-\frac{1}{\sqrt{2\pi\sigma_\ideal^2}}\int_A\e^{-\frac{x^2}{2\sigma_\ideal^2}}\dx\right| =0\,.
\end{equation}
By the Berry--Ess\'een theorem, the error in \eqref{eqn:CLT:ideal} is of the order $\mathcal{O}(1/\sqrt{N})$.

Obtaining a comparable statement for the interacting Bose gas has been the content of several works. For our model, one can show along the lines of \cite{rademacher2019_2} that
\begin{equation}\label{eqn:CLT:int}
\lim\limits_{N\to\infty}\left|\P_{\PsiNgs} (\cBN\in A)- \frac{1}{\sqrt{2\pi\sigma^2}}\int_A\e^{-\frac{x^2}{2\sigma^2}}\dx \right|=0
\end{equation} 
for 
\begin{equation}\label{eqn:sigma}
\sigma:=\norm{\nu}\,,\qquad\nu:=U_0qB\varphi +\overline{V_0}\overline{qB\varphi}\,,
\end{equation}
for $U_0$ and $V_0$ from \eqref{eqn:U_0:V_0} and $q$ as in \eqref{p:q}. In general, $\sigma$ and $\sigma_N$ differ by an error of order $\mathcal{O}(1)$. Hence, the interactions have a visible effect on the level of the CLT: they change the variance of the limiting Gaussian random variable.

The simplest way to understand this effect is via the characteristic functions of the random variables $\cBN^\ideal$ and $\cBN$, which are given by
\begin{equation}\label{eqn:char:fctn:ideal}
\phi_N^\ideal(s):=\lr{\varphi^{\otimes N},\e^{\i s\cBN^\ideal}\varphi^{\otimes N}}=\lr{\varphi,\e^{\frac{\i s}{\sqrt{N}} (B-\lr{\varphi,B\varphi})}\varphi}^N
\end{equation}
for the ideal gas, and by
\begin{equation}\label{eqn:char:fctn:int}
\phi_N^\gs(s):=\lr{\PsiNgs,\e^{\i\cBN s}\PsiNgs}
\end{equation}
for the interacting gas. To compute the inner products in \eqref{eqn:char:fctn:ideal} and \eqref{eqn:char:fctn:int}, one applies the map $\UNp$ from \eqref{intro:UNp} to the $N$-body states $\varphi^{\otimes N}$ and $\PsiN$. Since $\varphi^{\otimes N}$ is the pure condensate, $\UNp$ maps $\varphi^{\otimes N}$ onto the vacuum $\vac$ of the excitation Fock space, whereas $\UNp\PsiNgs=\BogUz^*\vac +\mathcal{O}(N^{-1/2})$ (see Lemma \ref{lem:known}). Conjugating $\cBN$ with $\UNp$ and, for the interacting gas case, with $\BogUz$, leads to the identities
\begin{eqnarray}
\phi_N^\ideal(s)&=&\lr{\Omega,\e^{\ad(\i s qB\varphi)-a(\i s qB\varphi)}\Omega}+\mathcal{O}(N^{-\frac12})
\;=\; \e^{-\frac12\norm{qB\varphi}^2s^2}+\mathcal{O}(N^{-\frac12})\,,\\
\phi_N^\gs(s)&=&\lr{\Omega,\BogUz\e^{\ad(\i s qB\varphi)-a(\i s qB\varphi)}\BogUz^*\Omega}+\mathcal{O}(N^{-\frac12})
\;=\; \e^{-\frac12\sigma^2s^2}+\mathcal{O}(N^{-\frac12})\,\label{eqn:phi_N:leading:order}
\end{eqnarray}
(see  Section \ref{subsec:proof:strategy} for the details). 
Since
\begin{equation}
\norm{qB\varphi}^2=\lr{\varphi,B(1-|\varphi\rangle\langle\varphi|)B\varphi}=\sigma_\ideal^2\,,
\end{equation}
the inverse Fourier transform leads to the Gaussian probability densities as in \eqref{eqn:CLT:ideal} and~\eqref{eqn:CLT:int}.

The mathematical derivation of quantum central limit theorems has first been studied in the 1970s in \cite{cushen1971,hepp1973} and was followed by many works in different settings, e.g., \cite{goderis1989,speicher1992, kuperberg2005, hayashi2006,jaksic2010, cramer2010}. For the ground state of an interacting $N$-body system, \eqref{eqn:CLT:int} was proven in \cite{rademacher2019_2} for interactions in the Gross--Pitaevskii regime. For the mean-field Bose gas, the corresponding dynamical problem was first studied in \cite{benarous2013}, where the authors consider the time evolution generated by $\HN$ of an initial product state. This was generalized in \cite{buchholz2014}  to $k$ one-body operators (corresponding to a multivariate setting), in \cite{rademacher2019} to singular interactions, and in \cite{rademacher2021} to $k$-body operators (corresponding to $m$-dependent random variables).

\subsection{No Gaussian central limit theorem for low-energy eigenstates}\label{subsec:no:clt}

So far, we have considered the situation where the interacting Bose gas is in its ground state. If, instead, it is in a low-energy eigenstate $\PsiNex\in\ENmex$ for $\eta>0$, the limiting distribution of the fluctuations is not Gaussian.
For example, if the first excited state $\Psi_N^{(1)}$ is contained in $\ENoex$, it satisfies $\UNp\Psi_N^{(1)}=\BogUz^*\ad(\xi)\vac+\mathcal{O}(N^{-1/2})$ for some normalized $\xi\in\fH$ (see Lemma~\ref{lem:known:expansion:Chi:eta}). In this case, we find that
\begin{equation}\label{eqn:no:CLT:int}
\lim\limits_{N\to\infty}\left|\P_{\Psi_N^{(1)}} (\cBN\in A)- \int_A b^{(1)}_\infty(x)\dx \right|=0\,,
\end{equation} 
where
\begin{equation}\label{eqn:prob:density:one:exc}
b_\infty^{(1)}(x):=\left(1+\frac{|\lr{\xi,\nu}|^2}{\sigma^2}\left(\frac{x^2}{\sigma^2}-1\right)\right)\frac{1}{\sqrt{2\pi\sigma^2}}\,\e^{-\frac{x^2}{2\sigma^2}}
\end{equation}
(see also \cite[Appendix A]{rademacher2019_2}).
The general case with $n$ excitations is treated in Proposition~\ref{thm:no:clt}.

\subsection{Edgeworth expansion for the product state}\label{subsec:edgeworth:ideal:gas}
For the case of i.i.d.\ random variables, one can go beyond the order $N^{-1/2}$ of the CLT and approximate the probability distribution of $\cBN^\ideal$  in an Edgeworth series, i.e., in a power series in powers of $N^{-1/2}$ which is determined by the cumulants of the distribution. We follow the discussion from \cite[Chapter~2]{field_ronchetti}.
The $\l$'th  cumulant of the distribution of $\cBN^\ideal$ is defined as
\begin{equation}\label{eqn:cumulants}
\kappa_\l[\cBN^\ideal]:=(-\i)^\l\left(\tfrac{\d}{\d s}\right)^\l\ln \phi_N^\ideal(s)\Big|_{s=0}
\end{equation}
for $\l\in\N$, and one easily verifies that
\begin{equation}\label{eqn:cumulants:N:dep}
\kappa_\l[\cBN^\ideal] = N^{1-\frac{\l}{2}}\kappa_\l[\tilde{B}]\,,
\end{equation}
where we abbreviated
\begin{equation}
\tilde{B}:=B-\lr{\varphi,B\varphi}\,.
\end{equation}
The first three cumulants coincide with the first three central moments; in particular, 
\begin{equation}\label{eqn:first:cumulants}
\kappa_1[\tilde{B}]=\bE_\varphi[\tilde{B}]=0\,,\qquad\kappa_2[\tilde{B}]=\Var_\varphi[\tilde{B}]=\sigma_\ideal^2\,.
\end{equation}
The basic idea of the Edgeworth series is to expand $\phi_N^\ideal$ around the characteristic function $\exp(- s^2\sigma_\ideal^2/2)$ of the corresponding Gaussian random variable. 
Since the $\l$'th cumulant is the $\l$'th coefficient in the Taylor expansion of $\ln\phi_N^\ideal(s)$ around zero, one (formally) computes with \eqref{eqn:first:cumulants}
\begin{eqnarray}
\phi_N^\ideal(s)&=&\e^{\ln \phi_N^\ideal(s)+\tfrac12 s^2\sigma_\ideal^2 }\,\e^{-\frac12 s^2\sigma_\ideal^2}\nonumber\\
&=&\exp\left\{\sum_{\l\geq 3}N^{-\frac{\l}{2}+1}\frac{\kappa_\l[\tilde{B}](\i s)^\l}{\l!}\right\}\e^{-\frac12 s^2\sigma_\ideal^2}\nonumber\\
&=&\left( 1+ N^{-\frac12}\frac{\kappa_3(\i s)^3}{3!} + N^{-1}\left(\frac{\kappa_4(\i s)^4}{4!}+\frac{\kappa_3^2(\i s)^6}{2\cdot(3!)^2}\right) +\dots \right)\e^{-\frac12 s^2\sigma_\ideal^2}\,,\qquad\label{eqn:expansion:phi_N:ideal}
\end{eqnarray}
where we abbreviated $\kappa_\l:=\kappa_\l[\tilde{B}]$.
Applying the inverse Fourier transform leads to a series expansion for the probability density $b_N^\ideal$ of the random variable $\cBN^\ideal$,
\begin{eqnarray}
b_N^\ideal(x)&=&\bigg(1
+N^{-\frac12}\frac{\kappa_3}{6\,\sigma_\ideal^3}H_3\left(\tfrac{x}{\sigma_\ideal}\right)\nonumber\\
&&\quad+N^{-1}\left(\frac{\kappa_4}{24\,\sigma_\ideal^4}H_4\left(\tfrac{x}{\sigma_\ideal}\right)+\frac{\kappa_3^2}{72\,\sigma_\ideal^6}H_6 \left(\tfrac{x}{\sigma_\ideal}\right) \right) 
+\dots\bigg)\frac{1}{\sqrt{2\pi}\sigma_\ideal}\e^{-\frac{x^2}{2\sigma_\ideal^2}}\,,
\label{eqn:expansion:b_N:ideal}
\end{eqnarray} 
where
\begin{equation}\label{def:Hermite}
H_\l(x):=\e^\frac{x^2}{2}\left(-\tfrac{\d}{\d x}\right)^\l \e^{-\frac{x^2}{2}}
\end{equation}
are the (Chebyshev-)Hermite polynomials, for example
\begin{subequations}
\begin{eqnarray}
H_2(x)&=&x^2-1\,,\\
H_3(x)&=&x^3-3x\,,\\
H_4(x)&=&x^4-6x^2+3\,,\\
H_6(x)&=&x^6-15x^4+45x^2-15\,.
\end{eqnarray}
\end{subequations}
The functions $H_j$ are polynomials of degree $j$ which are even/odd for $j$ even/odd.
The complete (formal) Edgeworth expansion is given by the formula
\begin{equation}\label{eqn:EE:full:formal}
b_N^\ideal(x)
=\left(1+\sum\limits_{\l\geq1}N^{-\frac{\l}{2}}\mathfrak{p}_\l^\ideal(x)\right)\frac{1}{\sqrt{2\pi}\sigma_\ideal}\e^{-\frac{x^2}{2\sigma_\ideal^2}}
\end{equation}
with
\begin{equation}
\mathfrak{p}_\l^\ideal(x)=\sum_{m=1}^\l \frac{H_{\l+2m}\left(\tfrac{x}{\sigma^\ideal}\right)}{\sigma_\ideal^{\l+2m}m!}\sum_{\substack{\bj\in\N^m\\|\bj|=\l}}\prod_{n=1}^m \frac{\kappa_{j_n+2}}{(j_{n}+2)!}\,.
\end{equation}
The $\l$'th Edgeworth polynomial $\mathfrak{p}_\l^\ideal$ is a polynomial of degree $3\l$ which is even/odd for $\l$ even/odd and whose coefficients depend on the cumulants of $\tilde{B}$ of order up to $\l+2$.
If the expansion is truncated after finitely many terms, the right hand side of \eqref{eqn:EE:full:formal} is in general no probability density since it may become negative for large values of $|x|$ and is not necessarily normalized. The Edgeworth expansion is thus a local approximation, which is good in the center of the distribution but can be inaccurate in the tails.

The expansion \eqref{eqn:EE:full:formal} was first formally derived by Chebyshev and Edgeworth in the end of the 20'th century, and the first  proof is due to Cram\'er. Under the assumption that all relevant moments of the distribution exist, the rigorous statement is usually formulated as an asymptotic expansion of the cumulative distribution function  or the probability density, with an error that is uniform in $x$, i.e.,
\begin{equation}
b_N^\ideal(x)
-\left(1+\sum\limits_{\l=1}^aN^{-\frac{\l}{2}}\mathfrak{p}_\l^\ideal(x)\right)\frac{1}{\sqrt{2\pi}\sigma_\ideal}\e^{-\frac{x^2}{2\sigma_\ideal^2}}=
\smallO{N^{-\frac{a}{2}}}
\end{equation}
(see, e.g., \cite{wallace1958,feller,petrov,field_ronchetti,hall} and the references therein).
In general, one cannot take the limit $a\to\infty$ since the series does usually not converge. Generalizations of Edgeworth expansions for i.i.d.\ random variables, for example to different statistics, the multivariate case or the situation when the leading order is not Gaussian, can be found in the literature mentioned above.

\subsection{Edgeworth expansion for the interacting gas}\label{sec_Edgeworth_int}

Let us consider the ground state $\PsiNgs$ of the interacting gas. Due to the dependence of the random variables, this situation is much more intricate than for the product state. In Theorem~\ref{thm:edgeworth}, we prove that the probability density $b_N$ of the random variable $\cBN$ with probability distribution determined by $\PsiN$ is given by
\begin{equation}
b_N(x)=\left(1+\sum_{j=1}^aN^{-\frac{j}{2}}\mathfrak{p}_j(x)\right)\frac{1}{\sqrt{2\pi\sigma^2}}\e^{-\frac{x^2}{2\sigma^2}} + \mathcal{O}(N^{-\frac{a+1}{2}})
\end{equation}
in the weak sense of \eqref{eqn:thm}. 
Let us provide a formal derivation of this result. As a consequence of the interactions, the cumulants
\begin{equation}\label{def:cumulants:int}
\kappa_\l^\gs[\cBN]:=(-\i\tfrac{\d}{\d s})^\l\ln\phi_N^\gs(s)\big|_{s=0}
=(-\i\tfrac{\d}{\d s})^\l\ln\lr{\PsiNgs,\e^{\i\cBN s}\PsiNgs}\big|_{s=0}
\end{equation}
do not have the cumulative property that would lead to the exact scaling behaviour \eqref{eqn:cumulants:N:dep}. Instead, each cumulant $\kappa_\l^\gs[\cBN]$ has a series expansion in powers of $1/N$, for example
\begin{subequations}\label{eqn:expansion:cumulants}
\begin{eqnarray}
\kappa^\gs_2[\cBN] &=& \kappa_{2;0} + N^{-1}\kappa_{2;1} + N^{-2}\kappa_{2;2} + \dots\,,\\
\kappa^\gs_3[\cBN] &=& N^{-\frac12} \kappa_{3;0} + N^{-\frac32}\kappa_{3;1}+ N^{-\frac52}\kappa_{3;2}+\dots\,,\\
\kappa^\gs_4[\cBN] &=& N^{-1}\kappa_{4;0} + N^{-2}\kappa_{4;1} +N^{-3}\kappa_{4;2} +\dots
\end{eqnarray}
\end{subequations}
with
\begin{equation}
\kappa_{2;0}=\sigma^2\,,\qquad 
\kappa_{3;0}=\alpha_3
\end{equation}
for $\sigma$ as in \eqref{eqn:sigma} and $\alpha_3$ as in \eqref{def:alpha}. 
Note that the leading order of $\kappa_\l^\gs[\cBN]$ for $\l=2,3,4$ is $N^{-\l/2+1}$, which is the scaling behaviour of the corresponding cumulant in the i.i.d.\ case. Moreover, only even/odd powers of $N^{-1/2}$ contribute for $\l$ even/odd. 

Proving \eqref{eqn:expansion:cumulants} in full generality for each $\l\geq 2$ would be extremely tedious, which is why we refrain from following that route for a proof of Theorem \ref{thm:edgeworth}. Assuming one could prove the (formal) identity
\begin{equation}
\kappa_\l^\gs[\cBN]=\sum_{\nu\geq 0}N^{-\frac{\l}{2}-\nu+1}\kappa_{\l;\nu} 
\end{equation}
for each $\l\geq 2$, a computation along the lines of \eqref{eqn:expansion:phi_N:ideal} (formally) yields 
\begin{align}
b_N(x)&=\bigg(1
+N^{-\frac12}\frac{\alpha_3}{6\,\sigma^3}H_3\left(\tfrac{x}{\sigma}\right)\nonumber\\
&\quad+N^{-1}\left(\frac{1}{2\sigma}H_2(\tfrac{x}{\sigma}) +
\frac{\kappa_{4;0}}{24\,\sigma^4}H_4\left(\tfrac{x}{\sigma}\right)+\frac{\kappa_{3;0}^2}{72\,\sigma^6}H_6 \left(\tfrac{x}{\sigma}\right) \right) 
+\dots\bigg)\frac{1}{\sqrt{2\pi}\sigma}\e^{-\frac{x^2}{2\sigma^2}}\,,
\label{eqn:expansion:b_N}
\end{align} 
which is consistent with the rigorous result obtained  in Theorem~\ref{thm:edgeworth}.

\section{Proofs}\label{sec:proofs}
\subsection{Preliminaries}

\subsubsection{Weyl operators}\label{subsec:pre:weyl}

As a preparation, we recall in this section the concept of Weyl operators (see, e.g., \cite{rodnianski2009}) and collect some of their well-known properties.
For any $f\in \fH$, the Weyl operator is defined as
\begin{equation}
W(f) := \e^{\ad(f) - a(f)}.
\end{equation}
It is unitary with $W^*(f) = W(-f)$ and satisfies the shift property
\begin{equation}\label{eqn:Weyl}
W^*(f) a(g) W(f) = a(g) + \lr{g,f}\,,\qquad W^*(f) \ad(g) W(f) = \ad(g) + \lr{f,g}
\end{equation}
for all $f,g \in \fH$. 
Conjugation with a Bogoliubov transformation $\BogU$, $\BogV=\left(\begin{smallmatrix}U & \overline{V}\\ V & \overline{U}\end{smallmatrix}\right)$, transforms a Weyl operator into another Weyl operator as
\begin{equation}\label{Bog_of_Weyl}
\BogU W(f) \BogU^* = W(g)\,,\qquad g:=Uf-\overline{Vf}\,.
\end{equation}
The Baker--Campbell--Haussdorff formula  yields
\begin{equation}\label{BCH}
W(f) = \e^{\ad(f)} \e^{-a(f)} \e^{-\frac12 \norm{f}^2}\,,
\end{equation}
which leads to the identity
\begin{equation}\label{Weyl_in_vacuum}
\lr{\Omega, W(f) \Omega} =\e^{-\frac12 \norm{f}^2}\,.
\end{equation}
The number operator transforms under a Weyl operator as
\begin{equation}
\Np W(f) = W(f) \left( \Np + \ad(f)+a(f) + \norm{f}^2 \right)\,,
\end{equation}
which leads to the following result:

\begin{lem}\label{lem:aux:W}
Let $b\in \frac{1}{2}\N_0$ and $f\in \fH$. Then there exists a constant $C(b)$ such that
\begin{equation}\label{Weyl_number_op}
\norm{(\Np+1)^bW(f)\bxi}\leq C(b) \left( \norm{(\Np+1)^b \bxi} + \norm{f}^{2b}\norm{\bxi}\right)
\end{equation}
for any $\bxi\in\Fock$.
\end{lem}
\begin{proof}
By unitarity of the Weyl operator,
\begin{eqnarray}
\norm{(\Np+1)^bW(f)\bxi}
&=&\norm{(\Np+1+\ad(f)+a(f)+\norm{f}^2)^b\bxi}\nonumber\\
&\leq& C(b) \left( \norm{(\Np+1)^b \bxi} + \norm{f}^{2b}\norm{\bxi}\right)
\end{eqnarray}
where we used the estimate $\norm{a^\sharp(f)\bxi}\leq\norm{f}\norm{(\Np+1)^\frac12\bxi}$ for $a^\sharp\in\{\ad,a\}$.
\end{proof}

\subsection{Strategy of proof}\label{subsec:proof:strategy}
In this section, we give an overview of the proof of our main result, Theorem~\ref{thm:edgeworth}. 
We will in the following always assume that Assumptions \ref{ass:V} to \ref{ass:cond} are satisfied and that $\PsiN\in\ENmex$ for some $\eta\in\N_0$ (see Definition \ref{def:ENmex}). Moreover, we will use the notation
$\Chi=\UNp\Psi$
and denote by $\Chi_n$ the coefficients of the asymptotic expansion \eqref{eqn:expansion:Chi:spectrum:eta}. As above, we will only indicate the dependence on $\eta$ in the notation where it is inevitable.
Our goal is to compute the quantity
\begin{equation}\label{eqn:E[g]}
\mathbb{E}_\PsiN[g(\cBN)]
=\lr{\PsiN,g(\cBN)\PsiN}
=\int_\R\ds\,\hat{g}(s)\phi_N(s)
\end{equation}
with
\begin{equation}\label{phi_N}
\phi_N(s)=\lr{\PsiN,\e^{\i s\cBN}\PsiN}
\end{equation}
for $\cBN$ as in \eqref{cBN} and $g:\R\to\C$ some integrable and sufficiently regular function.
As a first step, we use the excitation map $\UNp$ from \eqref{intro:UNp} to re-write the characteristic function as
\begin{equation}\label{eqn:E:FT:Chi}
\phi_N(s)
=\lr{\Chi,\e^{\i s\FockB}\Chi}\,,
\end{equation} 
where $\FockB$ denotes the operator on $\Fock$ defined by
\begin{equation}
\FockB:=\UNp\cBN\UNp^*\,.
\end{equation}
Applying the substitution rules \eqref{eqn:substitution:rules} and expanding the square roots $\sqrt{1-\Np/N}$ in $N^{-1}$ leads to the following asymptotic expansion (see Section \ref{subsec:proof:lem:B} for the proof):

\begin{lem}\label{lem:B}
We have
\begin{equation}\label{eqn:B:series}
\FockB=\sum_{\l=0}^a N^{-\frac{\l}{2}}\FockB_\l+N^{-\frac{a+1}{2}}\FockR_a\,,
\end{equation}
where
\begin{subequations}\label{eqn:def:FockB_l}
\begin{eqnarray}
\FockB_1&=&\d\Gamma(q\tilde{B}q)-B^{(1)}\,,\\
\FockB_{2\l}&=&c_{\l}\left(\ad(qB\varphi)\Np^{\l}+\Np^{\l}a(qB\varphi)\right)\qquad (\l\geq0)\,,\\
\FockB_{2\l+1}&=&-B^{(\l+1)}\qquad (\l\geq 1)
\end{eqnarray}
\end{subequations}
for $B^{(\l)}$ as in \eqref{eqn:B^(l)}, with $c_0=1$, $c_1=-1/2$, $c_\l=-(2\l-3)!!/(2^\l\l!)$ ($\l\geq 2$) and with 
\begin{equation}\label{eqn:estimate:R_a}
\norm{\FockR_a\bxi}\leq\onorm{B} C(a) \left( \norm{(\Np+1)^{a+1}\bxi} + \delta_{a,0} N^{-1/2}\norm{(\Np+1)^{3/2}\bxi} \right)
\end{equation}
for some constant $C(a)>0$ and any $\bxi\in\Fock$.
\end{lem}

Note that the estimate \eqref{eqn:estimate:R_a} is by far not optimal in the powers of $\Np$ except for $a=0$, which determines the largest power of $s$ in Proposition~\ref{prop:S}.
In combination with Duhamel's formula,
\begin{equation}\label{eqn:Duhamel:B:B_0}
\e^{\i s\FockB}
=\e^{\i s\FockB_0} + \i \int\limits_0^s\d\tau\e^{\i\tau\FockB}
\left(\FockB-\FockB_0\right)
\e^{\i(s-\tau)\FockB_0}\,,
\end{equation}
Lemma \ref{lem:B} leads to an expansion of $\e^{\i s\FockB}$. Together with the asymptotic series \eqref{eqn:expansion:Chi:spectrum:eta} for $\Chi$, this yields the following expansion of \eqref{eqn:E:FT:Chi}, which is proven in Section \ref{subsec:proof:prop:decomposition}:

\begin{proposition}\label{prop:decomposition}
For $\phi_N$ as defined in \eqref{phi_N}, it holds that
\begin{equation}\label{eqn:prop:decomposition}
\left|\phi_N(s)
-\sum_{j=0}^aN^{-\frac{j}{2}}\sum_{m=0}^j\sum_{n=0}^m\lr{\Chi_n,\FockT_{j-m}(s)\Chi_{m-n}}\right|
\leq N^{-\frac{a+1}{2}}\left(C(a) + \left|\cS_a(s)\right|\right)\,,
\end{equation}
where
\begin{subequations}\label{def:TSIJ}
\begin{eqnarray}
\FockT_0(s)&:=&\e^{\i s\FockB_0}\,,\\
\FockT_j(s)&:=&\sum_{k=1}^j\sum_{\substack{\bl\in\mathbb{N}^k\\|\boldsymbol{\l}|=j}}\FockI_\bl^{(k)} \qquad (j\geq1)\,,\\
\cS_a(s) &:=& \sum_{m=0}^a\sum_{n=0}^{a-m}\sum_{k=1}^{m+1}\sum_{\substack{\bl\in\N^{k-1}\\|\bl|\leq m}}\lr{\Chi_{a-m-n},\FockJ^{(k)}_{m;\bl}(s)\Chi_n} \label{S_a_S}
\end{eqnarray}
\end{subequations}
\begin{subequations}\label{def:I:J}
with
\begin{eqnarray}
\FockI_{\bl}^{(k)}(s)
&:=&\int^s_{\Delta_k}\d\btau \e^{\i\tau_k\FockB_0}\FockB_{\l_k}\e^{\i(\tau_{k-1}-\tau_k)\FockB_0}\FockB_{\l_{k-1}}\e^{\i(\tau_{k-2}-\tau_{k-1})}\mycdots\FockB_{\l_1}\e^{\i(s-\tau_1)\FockB_0}
\,,\\
\FockJ^{(k+1)}_{a;\bl}(s)
&:=& \int_{\Delta_{k+1}}^s\d\btau\e^{\i\tau_{k+1}\FockB}\FockR_{a-|\bl|}\e^{\i(\tau_k-\tau_{k+1})\FockB_0}\FockB_{\l_k}\e^{\i(\tau_{k-1}-\tau_k)\FockB_0}\FockB_{\l_{k-1}}
\mycdots \FockB_{\l_1}\e^{\i(s-\tau_1)\FockB_0}\qquad
\end{eqnarray}
\end{subequations}
for $\bl=(\l_1\mydots \l_k)\in\N^k$, $\FockB_\l$ and $\FockR_\l$ as in Lemma \ref{lem:B}, and where we used the notation
\begin{equation}
\int_{\Delta_j}^s\d\boldsymbol{\tau}:=\i^j\int\limits_0^s\d\tau_1\int\limits_0^{\tau_1}\d\tau_2\cdots\int\limits_0^{\tau_{j-1}}\d\tau_j \,.
\end{equation}
\end{proposition}

To control the remainders of the expansion, it is crucial that 
$\FockB_0=\ad(qB\varphi)+a(qB\varphi)$, 
hence 
\begin{equation}
\e^{\i \tau \FockB_0} = W(\i\tau qB\varphi)
\end{equation}
is a Weyl operator. Moreover, the operators $\FockR_\l$ and $\FockB_\l$ can be bounded by powers of the number operator. Hence, applying Lemma \ref{lem:aux:W} repeatedly and making use of the fact that moments of the number operator with respect to $\Chi_\l$ are bounded uniformly in $N$ (Lemma~\ref{lem:known:moments:Chil}) yields an estimate of the error $\cS_a(s)$ (see Section \ref{subsec:proof:prop:S} for a proof):

\begin{proposition}\label{prop:S}
The term $\cS_a(s)$ from \eqref{S_a_S} satisfies
\begin{equation}
|\cS_a(s)|
\leq C_B(a)\left(1+|s|^{3a+3}+N^{-\frac12}|s|^{3a+4}\right)
\end{equation}
where $C_B(a)\leq C(a) (1+\onorm{B}^{3a+4})$ for some constant $C(a)$.
\end{proposition}

The next step is to compute the coefficients in the expansion \eqref{eqn:prop:decomposition}, which is done in Section~\ref{subsec:proof:prop:computation}. Since an explicit evaluation to any order is too complex to obtain in full generality, we focus on the dependence of the coefficients on $s$:

\begin{proposition}\label{prop:computation}
For $\FockT_j$ as in Proposition \ref{prop:decomposition}, we have
\begin{equation}\label{eqn:prop:computation:eta}
\sum_{m=0}^j\sum_{n=0}^m\lr{\Chi_n,\FockT_{j-m}(s)\Chi_{m-n}}
= p_j^{(\eta)}(s)\e^{-\frac12 s^2\sigma^2}\,
\end{equation}
for $\sigma$ as in \eqref{eqn:sigma} and 
where
$p_j^{(\eta)}$ is a polynomial of degree $3j+2\eta$ with complex coefficients depending on  $\varphi$, $B$, $\BogVz$ and $\Theta^{(\eta)}_{\l,j}$. Moreover, $p_j^{(\eta)}$ is even/odd  for $j$ even/odd. 
\end{proposition}

For the ground state $\PsiNgs\in\ENzex$, an explicit computation of the leading and next-to-leading order of the approximation is still feasible and yields the following result (see Section \ref{subsec:proof:polynomials} for the details of the computation):

\begin{proposition}\label{prop:polynomials}
Let $\eta=0$. For $j=0,1$, the polynomials in \eqref{eqn:prop:computation:eta} are given by
\begin{equation}
p^{(0)}_0(s)=1\,,\qquad p^{(0)}_1(s)=-\frac{\i}{6} \alpha_3 s^3\,,
\end{equation} 
where
\begin{equation}\label{def:alpha}
\alpha_3 = 12 \Re\lr{\nu^{\otimes 3},\Theta^{(0)}_{1,3}} +\lr{\nu,\left(\Uz q\tB q\Uz^*+\overline{\Vz q\tB q\Vz^*}\right)\nu}  +4\Re\lr{\nu,\Uz q\tB q\Vz^*\overline{\nu}}\,,
\end{equation}
and where $\Theta^{(0)}_{1,3}$ is given in \cite[Appendix~B]{QF}.
\end{proposition}
Theorem \ref{thm:edgeworth} follows from Propositions \ref{prop:decomposition} to \ref{prop:polynomials}  by Fourier inversion (see Section \ref{subsec:proof:thm:edgeworth} for the proof). 
For excited states $\PsiNex\in\ENmex$ with $\eta>0$, we explicitly compute only the leading order polynomial. The proof of the following proposition is given in Section~\ref{app:no:clt}.

\begin{proposition}\label{thm:no:clt}
Let $\eta>0$ and denote the quasi-particle states by $\xi_1\mydots\xi_\eta\in L^2(\R^d)$, i.e.,
\begin{equation}\label{eqn:def:Chi_0^k}
\Chi_0= \BogUz^*\ad(\xi_1)\mycdots\ad(\xi_\eta)\vac\,.
\end{equation}
Then $\mathfrak{p}_0^\ex$ in Theorem \ref{thm:ex} is given by 
\begin{equation}
\mathfrak{p}_{0}^\ex(x) = \sum_{\l=0}^\eta c_{\eta,\l} \left(\frac{-\i}{\sigma}\right)^{2\ell} H_{2\ell}\Big(\frac{x}{\sigma}\Big)\,,
\end{equation}
with $H_k$ the $k$-th Hermite polynomial (as defined in \eqref{def:Hermite}) and where
\begin{equation}
c_{\eta,\l}
:=\frac{(-1)^{\l}}{(\eta-\ell)!((\l)!)^2}\sum_{\pi,\pi'\in\mathfrak{S}_\eta}
\prod_{j=1}^{\eta-\ell} \lr{\xi_{\pi'(j)},\xi_{\pi(j)}}
\prod_{j'=\eta-\l+1}^\eta \lr{\xi_{\pi'(j')},\nu} \lr{\nu,\xi_{\pi(j')}} \,.
\end{equation}
\end{proposition}

Note that $\xi_i=\xi_j$ for $i\neq j$ is admitted, and that the formula for $c_{n,\l}$ simplifies if the functions $\xi_j$ are orthonormal. For $\eta=1$, we recover \eqref{eqn:prob:density:one:exc}.

\subsection{Proofs of the propositions}

\subsubsection{Proof of Lemma \ref{lem:B}}\label{subsec:proof:lem:B}
We decompose $\cBN$ as
\begin{equation}
\cBN=\frac{1}{\sqrt{N}}\d\Gamma(\tilde{B})=\frac{1}{\sqrt{N}}\left(\d\Gamma(pBq)+\d\Gamma(qBp)+\d\Gamma(q\tilde{B}q)\right)
\end{equation}
with
\begin{equation}
\tilde{B}:=B-\lr{\varphi,B\varphi}\,.
\end{equation}
Note that
\begin{equation}
\sqrt{1-\frac{\Np}{N}} = \sum_{\l=0}^b c_\l N^{-\l}\Np^\l + N^{-(b+1)}\tilde{\FockR}_{2b}\,,\qquad
\norm{\tilde{\FockR}_{2b}\bxi}\leq C(b)\norm{\Np^{b+1}\bxi}
\end{equation}
for any $\bxi\in\Fock$ and $b\in\N_0$ and where $[\tilde{\FockR}_{2b},\Np]=0$. Besides, by Lemma \ref{lem:known:expansion:B},  there exists some $r_B(a)\in\R$ with $|r_B(a)|\leq C(a)\onorm{B}$ such that
\begin{equation}
\lr{\PsiN,B_1\PsiN}-\lr{\varphi,B\varphi}=\sum_{\l=1}^a N^{-\l}B^{(\l)}+N^{-(a+1)}r_B(a)\,.
\end{equation}
Consequently,
\begin{eqnarray}
\FockB&=&N^{-\frac12}\UNp\left(\d\Gamma(qBp)+\d\Gamma(pBq)+\d\Gamma(q\tilde{B}q)-N\left(\lr{\PsiN,B_1\PsiN}-\lr{\varphi,B\varphi}\right)\right)\UNp^*\nonumber\\
&=&\ad(qB\varphi)\sqrt{1-\frac{\Np}{N}}+\sqrt{1-\frac{\Np}{N}}a(qB\varphi)+N^{-\frac12}\d\Gamma(q\tilde{B}q)\nonumber\\
&&-\left(\sum_{\l=1}^{a}N^{-(\l-\frac12)}B^{(\l)} + N^{-(a+\frac12)}r_B(a)\right)\nonumber\\
&=& \sum_{\l=0}^a N^{-\frac{\l}{2}}\FockB_\l + N^{-\frac{a+1}{2}}\FockR_a
\end{eqnarray}
for $\FockB_\l$ as in \eqref{eqn:B:series} and where $\FockR_a$ satisfies \eqref{eqn:estimate:R_a}.
\qed

\subsubsection{Proof of Proposition \ref{prop:decomposition}}\label{subsec:proof:prop:decomposition}

From \eqref{eqn:B:series}, it follows that $\FockB-\FockB_0=N^{-1/2}\FockR_0$ with 
\begin{equation}
\FockR_0=\sum\limits_{\l=1}^b N^{-\frac{\l-1}{2}}\FockB_\l+N^{-\frac{b}{2}}\FockR_b \,.
\end{equation}
Hence, \eqref{eqn:Duhamel:B:B_0} implies that
\begin{eqnarray}
\e^{\i s\FockB}
&=&\e^{\i s\FockB_0} + N^{-\frac12}\int\limits_{\Delta_1}^s\d\btau\e^{\i\tau_1\FockB}
\left(\sum\limits_{\l_1=1}^a N^{-\frac{\l_1-1}{2}}\FockB_{\l_1}+N^{-\frac{a}{2}}\FockR_a\right)
\e^{\i(s-\tau_1)\FockB_0}\nonumber\\
&=&\e^{\i s\FockB_0} + \sum_{\l_1=1}^aN^{-\frac{\l_1}{2}}\int_{\Delta_1}^s\d\btau\e^{\i\tau_1\FockB_0}\FockB_{\l_1}\e^{\i(s-\tau_1)\FockB_0}\nonumber\\
&&+N^{-\frac{a+1}{2}} \bigg[ \int_{\Delta_1}^s\d\btau\e^{\i\tau_1\FockB}\FockR_a\e^{\i(s-\tau_1)\FockB_0}
+\sum_{\l_1=1}^a\int_{\Delta_2}^s\d\btau\e^{\i\tau_2\FockB}\FockR_{a-\l_1}\e^{\i(\tau_1-\tau_2)\FockB_0}\FockB_{\l_1}\e^{\i(s-\tau_1)\FockB_0}\bigg]\nonumber\\
&&+\sum_{\l_1=1}^a\sum_{\l_2=1}^{a-\l_1}N^{-\frac{\l_1+\l_2}{2}}\int_{\Delta_2}^s\d\btau\e^{\i\tau_2\FockB}\FockB_{\l_2}\e^{\i(\tau_1-\tau_2)\FockB_0}\FockB_{\l_1}\e^{\i(s-\tau_1)\FockB_0}\nonumber\\
&=&\dots\,,
\end{eqnarray}
which eventually leads to the expansion
\begin{equation}
\e^{\i s\FockB}=\sum_{j=0}^aN^{-\frac{j}{2}}\FockT_j(s) 
+ N^{-\frac{a+1}{2}}\sum_{k=1}^{a+1}\sum_{\substack{\bl\in\N^{k-1}\\|\bl|\leq a}}\FockJ^{(k)}_{a;\bl}(s)
\end{equation}
with $\FockT_j(s)$ and $\FockJ^{(k)}_{a;\bl}(s)$ as in \eqref{def:TSIJ} and \eqref{def:I:J}.
This implies \eqref{eqn:prop:decomposition} by \eqref{eqn:expansion:Chi:spectrum:eta} because
\begin{eqnarray}
\phi_N(s)
&=&\lr{\Chi-\sum_{n=0}^aN^{-\frac{n}{2}}\Chi_n,\e^{\i s\FockB}\Chi}
+\sum_{n=0}^aN^{-\frac{n}{2}}\lr{\Chi_n,\e^{\i s\FockB}\left(\Chi-\sum_{m=0}^{a-n}N^{-\frac{m}{2}}\Chi_m\right)}\nonumber\\
&&+\sum_{n=0}^a\sum_{m=0}^{a-n}N^{-\frac{n+m}{2}}\lr{\Chi_n,\e^{\i s\FockB}\Chi_m}\,.
\end{eqnarray}
\qed

\subsubsection{Proof of Proposition \ref{prop:S}}\label{subsec:proof:prop:S}

Recall from \eqref{def:TSIJ} that
$$\cS_a(s) = \sum_{m=0}^a\sum_{n=0}^{a-m}\sum_{\mu=0}^m\sum_{k=1}^{\mu+1}\sum_{\substack{\bl\in\N^{k-1}\\|\bl|=\mu}}\lr{\Chi_{a-m-n},\FockJ^{(k)}_{m;\bl}(s)\Chi_n}$$
with 
$$\FockJ^{(k)}_{m;\bl}(s)
= \int_{\Delta_{k}}^s\d\btau\e^{\i\tau_{k}\FockB}\FockR_{m-|\bl|}\e^{\i(\tau_{k-1}-\tau_{k})\FockB_0}\FockB_{\l_{k-1}}\e^{\i(\tau_{k-2}-\tau_{k-1})\FockB_0}\FockB_{\l_{k-2}}
\mycdots \FockB_{\l_1}\e^{\i(s-\tau_1)\FockB_0}\,.$$
By \eqref{eqn:estimate:R_a}, we find for any $\bxi,\bxi'\in\Fock$ that
\begin{subequations}\label{eqn:estimate:J}
\begin{align}
&\left|\lr{\bxi',\FockJ^{(k)}_{m;\bl}(s)\bxi}\right| \nonumber\\
&\quad\leq\left|\int_{\Delta_k}^s\d\btau\norm{\FockR_{m-|\bl|}\e^{\i(\tau_{k-1}-\tau_k)\FockB_0}\FockB_{\l_{k-1}}
\mycdots \FockB_{\l_1}\e^{\i(s-\tau_1)\FockB_0}\bxi}\norm{\bxi'}\right|\nonumber\\
&\quad\leq C(m) \onorm{B}\norm{\bxi'} \int_{[0,s]^k}\d\btau \Bigg[ \norm{(\Np+1)^{m-|\bl|+1}W(\i\delta_{k-1}f)\FockB_{\l_{k-1}}
\mycdots \FockB_{\l_1}W(\i\delta_0 f) \bxi} \label{eqn:estimate:J:1} \\
&\qquad\qquad\qquad\qquad\qquad + \delta_{m,|\bl|} N^{-1/2} 
\norm{(\Np+1)^{3/2} W(\i\delta_{k-1}f) \FockB_{\l_{k-1}} \mycdots \FockB_{\l_1}W(\i\delta_0 f) \bxi} \Bigg] \label{eqn:estimate:J:2}
\end{align}
\end{subequations}
where we used the notation $f=qB\varphi$  and abbreviated
\begin{equation}
\delta_{k-1}:=\tau_{k-1}-\tau_k\,\quad \delta_0:=s-\tau_1\,.
\end{equation}
By definition \eqref{eqn:def:FockB_l} of the operators $\FockB_\l$ and using Lemma \ref{lem:known:expansion:B} we find that
\begin{equation}\label{eqn:estimate:FockB_l}
\norm{(\Np+1)^b\FockB_\l\bxi}\leq C(\l,b)\onorm{B}\norm{(\Np+1)^{b+\gamma_\l}\bxi}\,,\quad
\gamma_\l=\begin{cases}
0 & \text{ if } \l\geq 3 \text{ odd}\\
1 & \text{ if }\l=1\\
\frac{\l+1}{2} & \text{ if } \l \text{ even}
\end{cases}
\end{equation}
for any $b\in\frac{1}{2}\N_0$.
With $|\delta_j|\leq |s|$ for all $j\in\{0\mydots k-1\}$ for $\btau\in[0,s]^k$,
Lemma \ref{lem:aux:W} and \eqref{eqn:estimate:FockB_l} imply 
\begin{align}
&\norm{(\Np+1)^bW(\i\delta_{k-1}f)\FockB_{\l_{k-1}} \bxi} \nonumber\\
&\quad\leq C(\bl,b) \left( \norm{(\Np+1)^b \FockB_{\l_{k-1}} \bxi} + (s\onorm{B})^{2b} \norm{\FockB_{\l_{k-1}}\bxi} \right) \nonumber\\
&\quad\leq C(\bl,b) \onorm{B} \left( \norm{(\Np+1)^{b+\gamma_{\l_{k-1}}} \bxi} + (s\onorm{B})^{2(b+\gamma_{\l_{k-1}})} \norm{\bxi} \right)\,.
\end{align}
Using this estimate repeatedly yields
\begin{align}
&\norm{(\Np+1)^b W(\i\delta_{k-1}f) \FockB_{\l_{k-1}} \mycdots \FockB_{\l_1}W(\i\delta_0 f) \bxi}\nonumber\\
&\qquad \leq C(\bl,b)(1+\onorm{B}^{k-1+2(b+\Gamma_\bl)})\left(1+|s|^{2(b+\Gamma_\bl)}\right)\norm{(\Np+1)^{b+\Gamma_\bl} \bxi}
\end{align}
where 
\begin{equation}
\Gamma_\bl:=\gamma_{\l_1}+\dots+\gamma_{\l_{k-1}}\,, \qquad 0\leq\Gamma_\bl\leq |\bl| 
\end{equation}
by definition   \eqref{eqn:estimate:FockB_l} of $\gamma_\l$.
Moreover, $k\leq|\bl|+1$, hence
\begin{subequations}
\begin{align}
\eqref{eqn:estimate:J:1}
&\leq C(m)(1+\onorm{B}^{|\bl|+2m+3})\norm{(\Np+1)^{m+1} \bxi}\norm{\bxi'}\left(1+|s|^{2m+|\bl|+3}\right)\,,\\
\eqref{eqn:estimate:J:2}
&\leq \delta_{m,|\bl|} C(m)N^{-\frac12}(1+\onorm{B}^{3m+4})\norm{(\Np+1)^{3/2+m} \bxi}\norm{\bxi'}\left(1+|s|^{3m+4}\right)\,.
\end{align}
\end{subequations}
By \eqref{lem:moments:Chil}, we conclude that
\begin{equation}
|\cS_a(s)| \leq C(a)(1+\onorm{B}^{3a+4})\left(1+|s|^{3a+3}+N^{-\frac12}|s|^{3a+4}\right)\,.
\end{equation}
\qed

\subsubsection{Proof of Proposition \ref{prop:computation}}\label{subsec:proof:prop:computation}
Let us introduce the abbreviation
\begin{equation}
\phi(g):=\ad(g)+a(g)
\end{equation}
for any $g\in\fH$, and denote as above
$f=qB\varphi$
and
$\nu=U_0f+\overline{V_0 f}$, 
for $U_0$ and $V_0$ as in \eqref{eqn:U_0:V_0}. Recall that $\sigma=\norm{\nu}$.
For an operator $A\in\cL(\fH)$, the relations \eqref{eqn:Weyl} for the Weyl operator yield
\begin{equation}
W(g)\d\Gamma(A)W(g)^*=\d\Gamma(A)-\phi(Ag)+\lr{g,Ag}\,,
\end{equation}
hence the operators $\FockB_\l$ from \eqref{eqn:def:FockB_l} transform as
\begin{subequations}
\label{eqn:FockB:conjugated}
\begin{eqnarray}
\e^{\i\tau\FockB_0}\FockB_{2\l+1}\e^{-\i\tau\FockB_0}
&=& \FockB_{2\l+1} \qquad (\l\geq1),\\
\e^{\i\tau\FockB_0}\FockB_1\e^{-\i\tau\FockB_0}&=& 
\left(\FockB_1-\tau\phi(\i q\tilde{B}qB\varphi)+\tau^2\lr{\varphi,Bq\tilde{B}qB\varphi}\right)\,,\\
\e^{\i\tau\FockB_0}\FockB_{2\l}\e^{-\i\tau\FockB_0}
&=&c_\l\bigg[\left(\ad(f)+\i\tau\norm{f}^2\right)\left(\Np-\tau\phi(\i f)+\tau^2\norm{f}^2\right)^\l\nonumber\\
&& + \left(\Np-\tau\phi(\i f)+\tau^2\norm{f}^2\right)^\l\left(a(f)-\i\tau\norm{f}^2 \right)\bigg]\,.
\end{eqnarray}
\end{subequations}
We summarize these expressions in the following way, keeping only track on the $\tau$-dependence and on the total number of creation/annihilation operators $a^\sharp$:

\begin{definition}[Equivalence classes]

Consider a self-adjoint polynomial of degree $j$ in $\Np$ and $a^\sharp$, i.e., an expression of the form
\begin{equation}\label{def:operator:polynomial}
\sum_{\substack{n,m\geq 0\\2n+m=j}}
\sum_{\nu=0}^{n} \sum_{\mu=0}^m
\sum_{\boldsymbol{k}\in\{-1,1\}^\mu}\Np^\nu\int\dx^{(\mu)}\xi_\mu(x^{(\mu)})a_{x_1}^{\sharp_{k_1}}\mycdots a_{x_\mu}^{\sharp_{k_\mu}}+\hc
\end{equation}
for some $\xi_\mu\in L^2(\R^{d\mu})$. Here, we used the notation $a^{\sharp_{-1}}:=a$ and $a^{\sharp_1}:=\ad$. 
\defit{
\item Two polynomials \eqref{def:operator:polynomial} are equivalent with respect to the relation $\sim$ iff  they have the same degree $j$ and the number of operator-valued distributions $a_x^\sharp$ in each summand is even/odd for  $j$ even/odd. We denote the representatives of the equivalence classes with respect to the  relation~$\sim$ by $\FockF_j$, i.e.,
\begin{equation}
\FockF_j\sim
\sum_{\substack{n,m\geq 0\\2n+m=j}}
\sum_{\nu=0}^{n} \sum_{\substack{0\leq\mu\leq m\\j+\mu \text{ even}}}
\sum_{\boldsymbol{k}\in\{-1,1\}^\mu}\Np^\nu\int\dx^{(\mu)}\xi_\mu(x^{(\mu)})a_{x_1}^{\sharp_{k_1}}\mycdots a_{x_\mu}^{\sharp_{k_\mu}}+\hc
\,.
\end{equation}

\item Two polynomials \eqref{def:operator:polynomial} are equivalent with respect to the relation $\sim_j$ iff  they have a degree $\leq j$ and the number of operator-valued distributions $a_x^\sharp$ in each summand is even/odd for  $j$ even/odd. We denote the representatives of the equivalence classes with respect to the  relation~$\sim_j$ by $\FockF_{\leq j}$, i.e.,
\begin{equation}
\FockF_{\leq j}\sim_j \FockF_{\tilde{j}}
\end{equation}
for any $\tilde{j}\leq j$.
When using the notation $\FockF_{\leq j}$, we will drop the index $j$ from $\sim_j$.
}
\end{definition}

With respect to these equivalence classes,
$\FockI^{(k)}_{\bl}(s)\sim\FockI^{(k)}_{\tilde{\bl}}(s)$ if $\bl$ and $\tilde{\bl}$ differ only by a permutation of indices. Moreover, $\FockI^{(k)}_{\bl}(s)$ is equivalent to the operator where $\int_{\Delta_j}^s\d\btau$ is replaced by $\int_{[0,s]^j}\d\btau$.
The identities \eqref{eqn:FockB:conjugated} yield
\begin{subequations}\label{eqn:FockB:conjugated:equiv}
\begin{eqnarray}
\tilde{\FockB}_{2\l+1}&:=&\int_0^s\e^{\i\tau\FockB_0}\FockB_{2\l+1}\e^{-\i\tau\FockB_0}\d\tau
\sim s \,\FockF_0 \,,\\
\tilde{\FockB}_1&:=&\int_0^s\e^{\i\tau\FockB_0}\FockB_1\e^{-\i\tau\FockB_0}\d\tau
\sim\sum_{q=1}^3s^q\,\FockF_{3-q}
\\
\tilde{\FockB}_{2\l}&:=&\int_0^s\e^{\i\tau\FockB_0}\FockB_{2\l}\e^{-\i\tau\FockB_0}\d\tau
\sim\sum_{q=1}^{2\l+2}s^q\, \FockF_{2\l+2-q}\,.
\end{eqnarray}
\end{subequations}
Consequently, for $|\bl|=j$,
\begin{equation}
\FockI^{(k)}_{\bl}(s) \sim \tilde{\FockB}_{\l_1}\tilde{\FockB}_{\l_2}\mycdots\tilde{\FockB}_{\l_k}\e^{\i s\FockB_0} 
\sim \tilde{\FockB}_1^{k_1}\tilde{\FockB}_2^{k_2}\mycdots\tilde{\FockB}_j^{k_j}\e^{\i s\FockB_0}
\end{equation}
where $(k_1\mydots k_j)\in\{0\mydots j\}^j$, $k_1+\dots+k_j=k$ and $\sum_{n=1}^j nk_n=j$.
From \eqref{eqn:FockB:conjugated:equiv}, one infers that
\begin{equation}
\tilde{\FockB}_{\l}^k\sim
\begin{cases}
s^k\, \FockF_0  & \text{ if } \l\geq 3 \text{ odd},\\[5pt]
\sum_{n=k}^{3k}s^n\,\FockF_{3k-n}& \text{ if } \l=1\,, \\[5pt]
\sum_{n=k}^{k(\l+2)}s^n\,\FockF_{k(\l+2)-n}& \text{ if } \l \text{ even}.
\end{cases}
\end{equation}
Using the notation 
$$k_\mathrm{odd}:=\sum_{\substack{3\leq q\leq j\\ q\text{ odd}}} k_q\,,\qquad j_\mathrm{odd}:=\sum_{\substack{3\leq q\leq j\\ q\text{ odd}}} q k_q
 $$
one computes
\begin{eqnarray}
\tilde{\FockB}_1^{k_1}\tilde{\FockB}_2^{k_2}\mycdots\tilde{\FockB}_j^{k_j}
&\sim&\left(\prod\limits_{\substack{3\leq q\leq j\\q \text{ odd}}}\tilde{\FockB}_q^{k_q}\right) 
\tilde{\FockB}_1^{k_1} \left(\prod\limits_{\substack{2\leq q\leq j\\q \text{ even}}}\tilde{\FockB}_q^{k_q}\right) \nonumber\\
&\sim&s^{k_\mathrm{odd}}\sum_{n_1=k_1}^{3k_1}s^{n_1}\FockF_{3k_1-n_1}\prod\limits_{\substack{2\leq q\leq j\\q \text{ even}}}\sum_{n_q=k_q}^{k_q(q+2)}s^{n_q}\FockF_{k_q(q+2)-n_q}\nonumber\\
&\sim&s^{k_\mathrm{odd}}\sum_{n=k-k_\mathrm{odd}}^{2k+j-(2k_\mathrm{odd}+j_\mathrm{odd})}s^n\FockF_{2k+j-(2k_\mathrm{odd}+j_\mathrm{odd})-n}\nonumber\\
&=&\sum_{n=k}^{2k+j-(k_\mathrm{odd}+j_\mathrm{odd})}s^n\FockF_{2k+j-(k_\mathrm{odd}+j_\mathrm{odd})-n}
\end{eqnarray}
and consequently
\begin{eqnarray}
\FockT_j(s)
&\sim&\sum_{k=1}^j\sum_{n=k}^{2k+j-(k_\mathrm{odd}+j_\mathrm{odd})}s^n\FockF_{2k+j-(k_\mathrm{odd}+j_\mathrm{odd})-n}\,\e^{\i s\FockB_0}\,.
\end{eqnarray}
Note that $k_\mathrm{odd}+j_\mathrm{odd}=\sum_{3\leq q\leq j \text{ odd}}(q+1)k_q$ is even, hence the power of $s$ and the degree of $\FockF$ sum up to an even/odd number if $j$ is even/odd. Moreover, the highest power of $s$ is attained for $k=j$ (where $k_\mathrm{odd}=j_\mathrm{odd}=0$), which corresponds to the term $\FockI^{(j)}_{(1,1\mydots 1)}(s)$.  
Hence, we conclude that
\begin{equation}\label{eqn:expansion:T_j}
\FockT_j(s)\sim
\left(\sum_{n=1}^{j-1} s^n\,\FockF_{\leq 3j-n} + \sum_{n=j}^{3j} s^n\,\FockF_{3j-n} \right)\,\e^{\i s\FockB_0}
\sim \sum_{n=1}^{3j} s^n\,\FockF_{\leq 3j-n} \,\e^{\i s\FockB_0}\,.
\end{equation}
Moreover, 
\begin{equation}\label{eqn:FockT:BT}
\BogUz\FockT_j(s)\BogUz^*\sim\sum_{\l=1}^{3j} s^\l\,\FockF_{\leq 3j-\l}\,W(\i s\nu)
\end{equation}
where we have used that 
$\BogUz\FockB_0\BogUz^* = \phi(\nu) $ and $\e^{\i s\phi(\nu)}=W(\i s\nu)$.
By \eqref{eqn:def:Theta:spectrum:eta}, we obtain
\begin{eqnarray}\label{chi_T_chi}
&&\hspace{-1.5cm}\lr{\Chi_n,\FockT_{j-m}(s)\Chi_{m-n}}\nonumber\\
&\sim&\sum\limits_{\substack{0\leq p\leq 3n+\eta\\p+n+\eta \text{ even}}}
\sum\limits_{\substack{0\leq q\leq 3(m-n)+\eta\\q+m-n+\eta \text{ even}}}
\sum_{\l=1}^{3(j-m)} s^\l
\int\dx^{(q+p)}\overline{\Theta_{n,p}^{(\eta)}(x_{q+1}\mydots x_{q+p})} \nonumber\\
&&
\times \Theta^{(\eta)}_{m-n,q}(x^{(q)})\lr{\Omega,a_{x_{q+1}}\mycdots a_{x_{q+p}} \FockF_{\leq 3(j-m)-\l}\,W(\i s\nu)\ad_{x_1}\mycdots\ad_{x_q}\Omega}.\qquad\quad
\end{eqnarray}
Using that
\begin{eqnarray}
W(\i s\nu)\ad_{x_1}\mycdots\ad_{x_q}\vac
=\e^{-\frac12s^2\sigma^2}(\ad_{x_1}+\i s\overline{\nu(x_1)})\mycdots (\ad_{x_q}+\i s\overline{\nu(x_q}))\e^{\i s\ad(\nu)}\vac\,,
\label{eqn:W:ad:Omega}
\end{eqnarray}
we find by permutation symmetry of $\Theta^{(\eta)}_{m-n,q}$ that
\begin{eqnarray}
&&\hspace{-1.5cm}\int\dx^{(q)}\Theta^{(\eta)}_{m-n,q}(x^{(q)})W(\i s\nu)\ad_{x_1}\mycdots\ad_{x_q}\vac\nonumber\\
&=&\e^{-\frac12s^2\sigma^2}\sum_{r=0}^q(\i s)^{q-r}\tbinom{q}{r}\int\dx^{(r)}\tilde{\Theta}^{(\eta)}_{m-n,q,r}(x^{(r)})\ad_{x_1}\mycdots\ad_{x_r}\e^{\i s\ad(\nu)}\vac
\end{eqnarray}
for $\tilde{\Theta}^{(\eta)}_{m-n,q,r}(x^{(r)})=\int\dx_{r+1}\mycdots\dx_q\overline{\nu(x_{r+1})}\mycdots\overline{\nu(x_q)}\Theta^{(\eta)}_{m-n,q}(x^{(q)})$.
The inner product in \eqref{chi_T_chi} is non-zero only if it contains equal numbers of creation and annihilation operators. Since the operators  $\FockF_{\leq 3(j-m)-\l}$ have been conjugated by a Bogoliubov transformation (see \eqref{eqn:FockT:BT}), they contain  at each degree of the polynomial all possible combinations of creation and annihilation operators. 
Hence, expanding $\e^{\i s\ad(\nu)}$ yields
\begin{eqnarray}
&&\hspace{-1cm}\int\dx^{(r)}\dx_{q+1}\mycdots\dx_{q+p}\tilde{\Theta}^{(\eta)}_{m-n,q,r}(x^{(r)})\overline{\Theta^{(\eta)}_{n,p}(x_{q+1}\mydots x_{p+q})}\nonumber\\
&&\qquad\times
 \lr{\Omega,a_{x_{q+1}}\mycdots a_{x_{q+p}}\FockF_{\leq 3(j-m)-\l} \,\ad_{x_1}\mycdots\ad_{x_r} \e^{\i s\ad(\nu)}\Omega}\nonumber\\
&=& \sum_{\nu=0}^{3(j-m)-\l} c^{(\eta)}_{\nu,j,m,n,\l,q,r}  s^{p+3(j-m)-\l-r-2\nu}
\end{eqnarray}
for some coefficients $c^{(\eta)}_{\nu,j,m,n,\l,q,r}\in\C$. In particular, there is a non-zero contribution from $\nu=0$ by \eqref{eqn:expansion:T_j}.
In summary, 
\begin{eqnarray}
&&\hspace{-1cm}\lr{\Chi_n,\FockT_{j-m}(s)\Chi_{m-n}}\nonumber\\
&\sim&\sum\limits_{\substack{0\leq p\leq 3n+\eta\\p+n+\eta \text{ even}}}
\sum\limits_{\substack{0\leq q\leq 3(m-n)+\eta\\q+m-n+\eta \text{ even}}}
\sum_{\l=1}^{3(j-m)} s^\l
\e^{-\frac12 \sigma^2 s^2}\sum_{r=0}^q (\i s)^{q-r}
\int\dx^{(r)}\dx_{q+1}\mycdots\dx_{q+p}\tilde{\Theta}^{(\eta)}_{m-n,q,r}(x^{(r)})\nonumber\\
&&\qquad\times \overline{\Theta^{(\eta)}_{n,p}(x_{q+1}\mydots x_{p+q})}\lr{\Omega,a_{x_{q+1}}\mycdots a_{x_{q+p}}\FockF_{\leq 3(j-m)-\l} \,\ad_{x_1}\mycdots\ad_{x_r} \e^{\i s\ad(\nu)}\Omega}\nonumber\\
&\sim&\e^{-\frac12 \sigma^2 s^2}
\sum\limits_{\substack{0\leq p\leq 3n+\eta\\p+n+\eta \text{ even}}}
\sum\limits_{\substack{0\leq q\leq 3(m-n)+\eta\\q+m-n+\eta \text{ even}}}
\sum_{\l=1}^{3(j-m)}
\sum_{r=0}^q 
\sum_{\nu=0}^{3(j-m)-\l} c^{(\eta)}_{\nu,j,m,n,\l,q,r}s^{p+q+3(j-m)-2(r+\nu)}\,.
\end{eqnarray}
Note that the highest power of $s$ is $3j+2\eta$ and that $p+q+3(j-m)-2(r+\nu)$ is even/odd when $3j$ is even/odd. This yields \eqref{eqn:prop:computation:eta} with 
\begin{equation}
p^{(\eta)}_j(s)=\sum_{\substack{0\leq k\leq 3j+2\eta\\k+j\text{ even}}}c^{(j,\eta)}_ks^k 
\end{equation}
for $c^{(j,\eta)}_k\in\C$ with $c^{(j,\eta)}_{3j+2\eta}\neq0$. 
\qed

\subsubsection{Proof of Proposition \ref{prop:polynomials}}\label{subsec:proof:polynomials}

From Propositions~\ref{prop:decomposition} and \ref{prop:S}, we know that
\begin{subequations}\label{eqn:next}
\begin{align}
\phi_N(s)
&= \e^{-\frac12 s^2\sigma^2} \nonumber\\
&\quad+\, \i N^{-\frac12} \int\limits_0^s\d\tau\lr{W(-\i sf)\Chiz, W(-\i\tau f) \d\Gamma(q\tB q) W(\i\tau f)\Chiz}\label{eqn:next:integral:term}\\
&\quad+\, N^{-\frac12} \Big( \lr{W(-\i sf) \Chiz, \Chi_1} + \lr{\Chi_1,W(\i sf)\Chiz}\Big) \label{eqn:next:mixed:terms} \\
&\quad-\, \i N^{-\frac12} B^{(1)} \int\limits_0^s\d\tau\lr{\Chiz, W(\i sf) \Chiz} \label{eqn:expectation:term}\\
&\quad+\cO(N^{-1}) \nonumber 
\end{align}
\end{subequations}
with $f=qB\varphi$ as above.
\medskip

\noindent\textbf{Computation of \eqref{eqn:next:integral:term}}.
As above, we abbreviate $\nu=\Uz q O \varphi + \Vzbar\,\overline{qO\varphi}$. With  \eqref{Bog_of_Weyl} and \eqref{BCH}, we find
\begin{align}
\eqref{eqn:next:integral:term} 
&= \i N^{-\frac12} \int\limits_0^s\d\tau\lr{W(-\i s\nu) \Omega, W(-\i \tau\nu) \BogUz\d\Gamma(q\tB q)\BogUz^* W(\i \tau\nu) \Omega} \nonumber\\
&= \i N^{-\frac12} \e^{-\frac12 s^2\sigma^2} \int\limits_0^s\d\tau\lr{\e^{-\i s\ad(\nu)} \Omega, W^*(\i \tau\nu) \BogUz\d\Gamma(q\tB q)\BogUz^* W(\i \tau\nu) \Omega}\,.
\end{align}
For any one-body operator $A$, any ONB $(\varphi_i)$ of $\fH_\perp$, and $g\in \fH_\perp$, we have
\begin{align}
W^*(g) \BogUz\d\Gamma(A)\BogUz^* W(g) &= \sum_{i,j} A_{ij} \Big( a(\overline{V_0\varphi_i}) + \lr{\overline{V_0\varphi_i},g} + \ad(U_0\varphi_i) + \lr{g,U_0\varphi_i} \Big) \nonumber\\
&\qquad \times \Big( a(U_0\varphi_j) + \lr{U_0\varphi_j,g} + \ad(\overline{V_0\varphi_j}) + \lr{g,\overline{V_0\varphi_j}} \Big)\,,
\end{align}
where we denoted $A_{ij}:=\lr{\varphi_i,A\varphi_j}$.
Consequently, expanding the exponential yields
\begin{align}\label{eqn:next:integral:term_computation}
\eqref{eqn:next:integral:term} 
&= \i N^{-\frac12} \e^{-\frac12 s^2\sigma^2} \sum_{i,j} (q\tB q)_{ij} \int\limits_0^s\d\tau\, \Big\langle \e^{-\i s\ad(\nu)} \Omega, \Bigg[ \ad(U_0\varphi_i) \ad(\overline{V_0\varphi_j}) \nonumber\\
&\quad + \Big( \lr{\overline{V_0\varphi_i},\i\tau\nu} + \lr{\i\tau\nu,U_0\varphi_i} \Big) \ad(\overline{V_0\varphi_j}) + \ad(U_0\varphi_i) \Big( \lr{U_0\varphi_j,\i\tau\nu} + \lr{\i\tau\nu,\overline{V_0\varphi_j}} \Big) \nonumber\\
&\quad + \lr{\overline{V_0\varphi_i},\overline{V_0\varphi_j}} + \Big( \lr{\overline{V_0\varphi_i},\i\tau\nu} + \lr{\i\tau\nu,U_0\varphi_i} \Big) \Big( \lr{U_0\varphi_j,\i\tau\nu} + \lr{\i\tau\nu,\overline{V_0\varphi_j}} \Big) \Bigg] \Omega \Big\rangle \nonumber\\
&= \i N^{-\frac12} \e^{-\frac12 s^2\sigma^2} \big( \tilde{c}_1 s + \tilde{c}_3 s^3 \big)\,,
\end{align}
where $\tilde{c}_1,\tilde{c}_3\in\R$ are given by
\begin{subequations}
\begin{eqnarray}
\tilde{c}_1 &=& \Tr(\Vz q\tB q\Vz^*)\,,\\
\tilde{c}_3 &=& -\tfrac{1}{6}\Big(\lr{\nu,\Uz q\tB q\Uz^*\nu}  +\lr{\overline{\nu},\Vz q\tB q\Vz^*\overline{\nu}}\Big) -\tfrac{2}{3}\Re\lr{\nu,\Uz q\tB q\Vz^*\overline{\nu}}\,.
\end{eqnarray}
\end{subequations}

\noindent\textbf{Computation of \eqref{eqn:next:mixed:terms}}.
Using that
\begin{equation}
\Chi_1 = \BogUz^* \left( \int\dx\Theta^{(0)}_{1,1}(x)\ad_x|\Omega\rangle
+\int\dx^{(3)}\Theta^{(0)}_{1,3}(x^{(3)})\ad_{x_1}\ad_{x_2}\ad_{x_3}|\Omega\rangle\ \right)
\end{equation}
by Lemma~\ref{lem:known:expansion:Chi:eta}, one computes 
\begin{equation}
\eqref{eqn:next:mixed:terms}
=\i N^{-\frac12}\e^{-\frac12 s^2\sigma^2}\left(2\Re\lr{\nu,\Theta^{(0)}_{1,1}}s-2\Re\lr{\nu^{\otimes 3},\Theta^{(0)}_{1,3}}s^3\right)\,.
\end{equation}

\noindent\textbf{Computation of \eqref{eqn:expectation:term}}. We find, using first \eqref{Weyl_in_vacuum} and then Lemma~\ref{lem:known:expansion:B}, that
\begin{align}\label{eqn:expectation:term_computation}
\eqref{eqn:expectation:term} &= - \i N^{-1/2} \e^{-\frac12 s^2\sigma^2} s B^{(1)} \nonumber\\
&= - \i N^{-\frac12} \e^{-\frac12 s^2\sigma^2} s \left( \lr{\Chiz, \phi(qB\varphi) \Chi_1} +\lr{\Chi_1, \phi(qB\varphi) \Chiz} + \lr{\Chiz,\d\Gamma(q \tB q) \Chiz}\right) \nonumber\\
&= - \e^{-\frac12 s^2\sigma^2} s \, \frac{\d}{\d s} \Big( \eqref{eqn:next:integral:term} + \eqref{eqn:next:mixed:terms} \Big)\Big|_{s=0} \nonumber\\
&= -\i N^{-\frac12} \e^{-\frac12 s^2\sigma^2} s \left(\tilde{c}_1 + 2\Re \lr{\Theta^{(0)}_{1,1},\nu} \right). 
\end{align}
This concludes the proof of Proposition \ref{prop:polynomials}.
\qed

\subsection{Proof of Theorem \ref{thm:edgeworth}}\label{subsec:proof:thm:edgeworth}

Combining Propositions \ref{prop:decomposition}, \ref{prop:computation} and \ref{prop:S}, we find that
\begin{eqnarray}
\bigg|\phi_N(s)
-\sum_{j=0}^aN^{-\frac{j}{2}}p_j^{(\eta)}(s)\e^{-\frac12 s^2\sigma^2}
\bigg| 
\leq C_B(a)\left(N^{-\frac{a+1}{2}}(1+|s|^{3a+3})+N^{-\frac{a+2}{2}}|s|^{3a+4}\right)\,.
\end{eqnarray}
Consequently, by \eqref{eqn:E[g]},
\begin{eqnarray}\bigg|\mathbb{E}[g(\cBN)]-
\sum_{j=0}^aN^{-\frac{j}{2}} \int_\R\ds\,\hat{g}(s)p_j^{(\eta)}(s)\e^{-\frac12 s^2\sigma^2}\bigg|\leq C_B(g,a)N^{-\frac{a+1}{2}} 
\end{eqnarray}
because $\hat{g}\in L^1(\R,(1+|s|^{3a+4})$. Finally, Plancherel's theorem implies that
\begin{eqnarray}
\int\ds\hat{g}(s)s^k\e^{-\frac12 s^2\sigma^2}
&=&\frac{1}{\sqrt{2\pi\sigma^2}}\int \dx g(x)\left(\i\frac{\d}{\dx}\right)^k\e^{-\frac{x^2}{2\sigma^2}}\nonumber\\
&=&\frac{1}{\sqrt{2\pi\sigma^2}}\left(\frac{-\i}{\sigma}\right)^k\int g(x) H_k\left(\frac{x}{\sigma}\right) \e^{-\frac{x^2}{2\sigma^2}}\,,\label{plancherel}
\end{eqnarray}
where $H_k(x)$ is the $k$-th Hermite polynomial as defined in \eqref{def:Hermite}.
This yields \eqref{eqn:thm} with polynomials $\mathfrak{p}_{j}(x)$ of degree $3j + 2\eta$ in $x\in\R$ which are even/odd for $j$ even/odd. Note that the coefficients of the $\mathfrak{p}_j$ must be real-valued because $\mathbb{E}[g(\cBN)]\in\R$ for real-valued $g$.
\qed

\subsection{Proof of Proposition \ref{thm:no:clt}}\label{app:no:clt}
We consider $\Chiz\in\ENmex$ for some $\eta>0$. The leading order of $\phi_N^{(\eta)}(s)$ is given by $\lr{\Chi_0,\e^{\i s\FockB_0}\Chi_0}$ and can be computed similarly to Propositions~\ref{prop:computation} and \ref{prop:polynomials}. Using \eqref{eqn:def:Chi_0^k} and abbreviating
\begin{equation}
\sigma_j:=\lr{\nu,\xi_j}\,,
\end{equation}
we find that
\begin{eqnarray}
\lr{\Chi_0,\e^{\i s\FockB_0}\Chi_0}
&=&\lr{\Omega,a(\xi_1)\mycdots a(\xi_\eta) W(\i s\nu)\ad(\xi_1)\mycdots \ad(\xi_\eta)\Omega}\nonumber\\
&=&\e^{-\frac12 s^2\sigma^2}\sum_{\l=0}^\eta s^{2(\eta-\l)}\frac{(-1)^{\eta-\l}}{\ell!((\eta-\l)!)^2}\sum_{\pi\in\mathfrak{S}_\eta}
\sigma_{\pi(\l+1)}\mycdots\sigma_{\pi(\eta)}\nonumber\\
&&\times\lr{\Omega,a(\xi_1)\mycdots a(\xi_\eta)
\ad(\xi_{\pi(1)})\mycdots\ad(\xi_{\pi(\l)})\ad(\nu)^{\eta-\l}\Omega}\nonumber\\
&=:&\e^{-\frac12 s^2\sigma^2}\sum_{\l=0}^\eta c_{\eta,\eta-\l} s^{2(\eta-\l)}\,,
\end{eqnarray}
where $\mathfrak{S}_\l$ denotes the set of permutations of $\l$ elements. To compute the coefficients $c_{\eta,\l}$, let us introduce the notation 
\begin{equation}
\zeta_j:=\begin{cases}
\xi_{\pi(j)} & j=1\mydots \l\,\\
\nu & j=\l+1\mydots \eta
\end{cases}
\end{equation}
and $I_\eta:=\{1\mydots \eta\}$.
Since
\begin{eqnarray}
&&\hspace{-1.5cm}\lr{\Omega,a(\xi_1)\mycdots a(\xi_\eta) \ad(\zeta_1)\mycdots\ad(\zeta_\eta)\Omega}\nonumber\\
&=&\sum_{j=1}^\eta \lr{\xi_\eta,\zeta_j}\lr{\Omega,a(\xi_1)\mycdots a(\xi_{\eta-1})\Big(\prod_{\mu\in I_\eta\setminus\{j\}}\ad(\zeta_j)\Big)\Omega}\nonumber\\
&=&\sum_{\pi'\in\mathfrak{S}_\eta}\lr{\xi_{\pi'(1)},\zeta_1}\mycdots\lr{\xi_{\pi'(\eta)},\zeta_\eta}\nonumber\\
&=&\sum_{\pi'\in\mathfrak{S}_\eta}\lr{\xi_{\pi'(1)},\xi_{\pi(1)}}\mycdots\lr{\xi_{\pi'(\l)},\xi_{\pi(\l)}}
\overline{\sigma_{\pi'(\l+1)}}
\mycdots \overline{\sigma_{\pi'(\eta)}} \,,
\end{eqnarray}
the coefficients $c_{\eta,\l}$ are given by
\begin{equation}\begin{split}
c_{\eta,\l}
=&\frac{(-1)^{\l}}{(\eta-\ell)!((\l)!)^2}\sum_{\pi,\pi'\in\mathfrak{S}_\eta}
\lr{\xi_{\pi'(1)},\xi_{\pi(1)}}\mycdots\lr{\xi_{\pi'(\eta-\l)},\xi_{\pi(\eta-\l)}}
\sigma_{\pi(\eta-\l+1)}\mycdots\sigma_{\pi(\eta)}
 \\
&\times\overline{\sigma_{\pi'(\eta-\l+1)}}
\mycdots\overline{\sigma_{\pi'(\eta)}}
\,.
\end{split}\end{equation}
This concludes the proof by \eqref{plancherel}.
\qed

\subsection*{Conflict of interest}
On behalf of all authors, the corresponding author states that there is no conflict of interest.

\section*{Acknowledgements}
It is a pleasure to thank Martin Kolb, Simone Rademacher, Robert Seiringer and Stefan Teufel for helpful discussions. Moreover, we thank the referee for many constructive  comments.
L.B.\ gratefully acknowledges funding from the German Research Foundation within the Munich Center of Quantum Science and Technology (EXC~2111) and from the European Union’s Horizon 2020 research and innovation programme under the Marie Sk{\textl}odowska-Curie Grant Agreement No.~754411. 
We thank the Mathematical Research Institute Oberwolfach, where part of this work was done, for their hospitality.


\bibliographystyle{abbrv}
    \bibliography{bib_file}

\end{document}